\documentclass[11pt]{article}
\usepackage{indentfirst}
\usepackage{float}

\usepackage{tikz}
\usetikzlibrary{shapes, fit, backgrounds, calc, positioning}
\usetikzlibrary{arrows.meta,decorations.pathreplacing}

\usepackage[utf8]{inputenc}
\usepackage{amsmath}
\usepackage{amssymb}
\usepackage{amsthm}
\usepackage{cite}
\usepackage{booktabs}
\usepackage{xcolor}

\newtheorem{theorem}{Theorem}

\newtheorem{lemma}{Lemma}
\newtheorem{example}{Example}

\newcommand\blfootnote[1]{%
  \begingroup
  \renewcommand\thefootnote{}\footnote{#1}%
  \addtocounter{footnote}{-1}%
  \endgroup
}

\newcommand{\Fq}{\mathbb{F}_q}
\newcommand{\F}{\mathbb{F}}
\newcommand{\cE}{\mathcal{E}}
\newcommand{\cH}{\mathcal{H}}
\newcommand{\cT}{\mathcal{T}}
\newcommand{\cU}{\mathcal{U}}
\newcommand{\cB}{\mathcal{B}}
\newcommand{\cI}{\mathcal{I}}
\newcommand{\cA}{\mathcal{A}}
\newcommand{\cF}{\mathcal{F}}
\newcommand{\cR}{\mathcal{R}}
\newcommand{\Eu}{\mathcal{E}_{\mathrm u}}
\newcommand{\Es}{\mathcal{E}_{0}}
\newcommand{\Ws}{W_{\Sigma}}

\tikzset{
  party/.style={circle,draw,thick,minimum size=7mm,inner sep=0pt,fill=white},
  dest/.style={circle,draw,thick,double,double distance=1.2pt,minimum size=7mm,inner sep=0pt,fill=white},
  serv/.style={rectangle,draw,thick,minimum width=9mm,minimum height=7mm,inner sep=1pt,fill=white},
  thin link/.style={gray!70},
  hon link/.style={very thick},
  up link/.style={gray!70,dashed,thick},
  coal/.style={rounded corners=10pt,fill=red!12,draw=red!50,thick}}

\title{The Capacity of Interactive Secure Summation}
\author{Hua Sun}
\date{}

\begin{document}
\maketitle

\blfootnote{
Hua Sun (email: hua.sun@unt.edu) is with the Department of Electrical Engineering at the University of North Texas.}

\begin{abstract}
In secure summation, $K$ parties each hold an input, and the sum of the inputs is computed at one party, securely against at most $T$ colluding parties. Communication typically takes one hop: each party sends one message directly to the party that computes the sum, using correlated randomness that is independent of the inputs and held in advance. We consider interactive secure summation, where every pair of parties is connected by a secure link, so that each party can send messages to every other party, each party has only independent local randomness, and the messages may depend on the inputs and on earlier messages over multiple rounds. We characterize the capacity region, the set of achievable rates of communication, counting all messages exchanged among the parties over all rounds, and of the randomness of every party, counted in symbols per symbol of the sum. For $0\le T\le K-2$, the optimal communication rate is $K(K-1)/(K-T)$, and the randomness rates form a polyhedron with, up to a relabeling of the parties, a single extreme point: $T$ parties have the randomness rates $(K+T-2t)/(K-T)$, $t=1,\ldots,T$, and the others none. In particular, the optimal total randomness rate is $T(K-1)/(K-T)$. Linear schemes achieve every extreme point together with the optimal communication rate. The same techniques give the capacity region when every party recovers the sum and when the sum is computed at a server without input.
\end{abstract}

\allowdisplaybreaks
\emergencystretch=3em
\section{Introduction}\label{sec:intro}

In secure summation~\cite{ZS22,ZS24}, $K$ parties each hold an input vector $W_k$ of $L$ symbols from a finite field $\Fq$ and the sum of the inputs is computed at one party, the \emph{destination}, securely against at most $T$ colluding parties. Communication typically takes one hop: each party sends one message directly to the destination, using correlated randomness held in advance. We consider \emph{interactive secure summation}: every pair of parties is connected by a secure link, so that each party can send messages to every other party, each party has only independent local randomness, and every message may depend on the sender's input, its randomness and the messages it has received. Interaction is generally hard to handle in network information theory: each transmitted symbol may depend on all symbols received before, as with feedback and two-way communication, and capacity results for interactive network models are rare~\cite{ElGamalKim11}. We study the capacity region of interactive secure summation: the achievable communication rate, the total number of transmitted symbols per symbol of the sum, together with the randomness rate of every party, the amount of local randomness that the party generates per symbol of the sum, as the input length $L$ grows and the input symbols are coded jointly. Since the local randomness of different parties is independent, any randomness that the parties come to share is generated by one of them and counted there, and the total randomness rate is the sum of the randomness rates of all parties.

The two extreme cases of the capacity problem, $T=1$ and $T=K-2$, have been solved implicitly in theoretical computer science, where the problem is called private computation of the sum~\cite{CK93,KR98,BDPV99,KM97,GR05,GIS22,FY92}, although the performance metrics there are different: each party holds a single bit or group element, and the costs are the numbers of messages and of random bits for one computation. The optimal schemes for these two extreme cases, applied to each symbol of the inputs, achieve the optimal communication rate and the optimal total randomness rate, and Figure~\ref{fig:solved} shows them for $K=4$. For $T=1$, the destination sends a randomness symbol to one party, and the running sum of the inputs plus this symbol passes through the other parties and returns to the destination; this uses $K$ messages and $1$ randomness symbol, and both numbers are necessary~\cite{CK93,KR98}. For $T=K-2$, every pair of parties other than the destination shares a randomness symbol generated by one of the two, and each of these parties sends its input plus noise from the randomness symbols it shares to the destination; this uses $\binom K2$ messages and $\binom{K-1}2$ randomness symbols, and both numbers are necessary~\cite{CK93,BDPV99}. These works count only the total randomness and do not consider the randomness of each party; it turns out that the schemes of Figure~\ref{fig:solved} are optimal for this metric as well, since their allocations, $1$ randomness symbol at the destination in (a) and $2,1,0,0$ randomness symbols at parties $1,2,3,4$ in (b), are extreme points of the polyhedron of randomness rates described below (Section~\ref{sec:results}). For $2\le T\le K-3$, the capacity problem is open: the theoretical computer science works in this regime~\cite{KM97,GR05,GIS22,FY92} bound the number of messages or random bits for one or a few input bits, up to constant or logarithmic factors, and do not settle the capacity, the cost per symbol of the sum for long inputs. The relation between our capacity results and these results is discussed further in Appendix~\ref{app:tcs}.

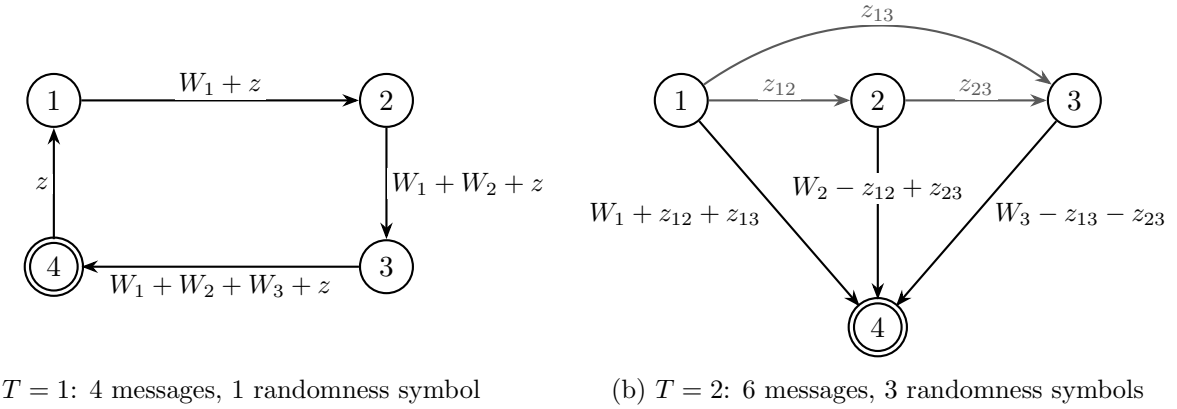
\begin{figure}[t]
\centering
\begin{tikzpicture}[msg/.style={-{Stealth[length=2.2mm]},thick},lab/.style={font=\small,fill=white,inner sep=1.5pt}]
\begin{scope}
\node[party] (a1) at (-2.2,1.1) {$1$};
\node[party] (a2) at (2.2,1.1) {$2$};
\node[party] (a3) at (2.2,-1.1) {$3$};
\node[dest]  (a4) at (-2.2,-1.1) {$4$};
\draw[msg] (a4) -- node[lab,left] {$z$} (a1);
\draw[msg] (a1) -- node[lab,above] {$W_1+z$} (a2);
\draw[msg] (a2) -- node[lab,right] {$W_1+W_2+z$} (a3);
\draw[msg] (a3) -- node[lab,below=1pt] {$W_1+W_2+W_3+z$} (a4);
\node[font=\small] at (0,-2.75) {(a) $T=1$: $4$ messages, $1$ randomness symbol};
\end{scope}
\begin{scope}[xshift=8.7cm]
\node[party] (b1) at (-2.6,1.1) {$1$};
\node[party] (b2) at (0,1.1) {$2$};
\node[party] (b3) at (2.6,1.1) {$3$};
\node[dest]  (b4) at (0,-1.9) {$4$};
\draw[msg,gray!70!black] (b1) -- node[lab,above] {$z_{12}$} (b2);
\draw[msg,gray!70!black] (b2) -- node[lab,above] {$z_{23}$} (b3);
\draw[msg,gray!70!black] (b1) to[bend left=35] node[lab,above] {$z_{13}$} (b3);
\draw[msg] (b1) -- node[font=\small,left=3pt,pos=0.5] {$W_1+z_{12}+z_{13}$} (b4);
\draw[msg] (b2) -- node[lab,pos=0.35] {$W_2-z_{12}+z_{23}$} (b4);
\draw[msg] (b3) -- node[font=\small,right=3pt,pos=0.5] {$W_3-z_{13}-z_{23}$} (b4);
\node[font=\small] at (0,-2.75) {(b) $T=2$: $6$ messages, $3$ randomness symbols};
\end{scope}
\end{tikzpicture}
\caption{The schemes for one symbol of the inputs; $W_k$ denotes the input symbol of party $k$. Double circle: destination. In (a), $z$ is a randomness symbol generated by party $4$; in (b), $z_{ij}$ is a randomness symbol generated by party $i$ and sent to party $j$.}
\label{fig:solved}
\end{figure}

Our main result settles the exact capacity of interactive secure summation for all $K$ and $T$,\footnote{The cases $T=K-1$ and $T=K$ reduce to $T=K-2$. The additional sets of colluding parties that include the destination are trivial, since they can compute all inputs from the sum and their own inputs; the parties other than the destination must still not learn its input, which holds in the schemes of Section~\ref{sec:scheme} because this input never enters a message.} including the open regime $2\le T\le K-3$ (Section~\ref{sec:results}). For $0\le T\le K-2$, the capacity region is the product of the set of communication rates at least $K(K-1)/(K-T)$ and a polyhedron of randomness rates, which consists of the vectors that are at least a convex combination of its extreme points. Up to a relabeling of the parties, the polyhedron has a single extreme point: $T$ distinct parties $p_1,\ldots,p_T$ have the randomness rates $(K+T-2t)/(K-T)$, $t=1,\ldots,T$, and the others none; the rates add up to the optimal total randomness rate $T(K-1)/(K-T)$. Since the destination plays a different role, the extreme points fall into $T+1$ \emph{types}: the destination is one of $p_1,\ldots,p_T$, or none of them. Linear schemes achieve every extreme point together with the optimal communication rate, so the optimal communication rate and every minimal allocation of the randomness are achieved simultaneously. At $T=1$ and $T=K-2$, the values agree with those of the two extreme cases.

For $(K,T)=(5,2)$, the smallest case with $2\le T\le K-3$, the optimal communication rate is $20/3$ and, up to a relabeling of the parties, the polyhedron of randomness rates has the single extreme point $(5/3,1,0,0,0)$: two parties with randomness rates $5/3$ and $1$, and the other three with none. The destination, party $5$, has randomness rate $5/3$, $1$ or $0$, which gives three types. Figure~\ref{fig:scheme} shows a scheme at $(5/3,1,0,0,0)$ with destination party $5$, worked out in Example~\ref{ex:scheme}. It is linear and has two decoupled phases.
\begin{itemize}
\item \emph{Phase 1: zero-sum noise} (Figure~\ref{fig:scheme}(a)). The input of each party has three symbols, so there are three sums, and for each sum the parties prepare noise symbols, one at every party, that add up to zero. An \emph{additive sharing of zero}, whose shares at any $4$ parties are independent and uniform, would take $4$ randomness symbols for each sum. Since at most $T=2$ parties collude, two sharings suffice: parties $1$ and $2$ each deal one, and the noise of every party for each sum is a linear combination of its two shares, with coefficients from a Cauchy matrix that differ across the sums. This takes $8$ randomness symbols for all three sums, does not use the inputs and does not depend on which party is the destination.
\item \emph{Phase 2: routing} (Figure~\ref{fig:scheme}(b)--(d)). Every party adds its noise to its input symbols, and for each of the three sums the \emph{noisy input symbols} are added up along a tree to the destination, a different tree for each sum, so that together with phase 1 every link carries $2$ symbols.
\end{itemize}
The noise cancels in each sum, which gives correctness. Security follows by counting: given what a coalition knows, all that it receives is determined by a certain number of symbols, and the noise in these symbols has full rank, that is, it occupies all of their dimensions, so they are independent and uniform and reveal nothing beyond the sum. Parties $1$ and $2$ generate $5$ and $3$ randomness symbols for the three sums, the rates $5/3$ and $1$, with total $8/3$. Since the role of the destination enters only phase 2, the same scheme works wherever the destination is. The schemes (Section~\ref{sec:scheme}) extend that of Figure~\ref{fig:scheme} to all $K$ and $T$, with $T$ parties that hold randomness, and achieve every extreme point of the polyhedron of randomness rates together with the optimal communication rate.

The converse (Section~\ref{sec:converse}) must hold for all interactive protocols, including nonlinear ones. It rests on two observations: the messages on a party's links, together with the sum, determine its input, since the inputs are independent and information about an input can leave its party only through these messages; and security forces the inputs of any $K-T$ \emph{honest} parties, beyond their sum, to remain uncertain to the other $T$ parties, if these collude, so this uncertainty must be carried by the messages among the honest parties. Averaging over all sets of $K-T$ honest parties bounds the communication, and comparing entropies bounds the total randomness. Restricting to the coalitions that contain a given set of parties, whose randomness they know, bounds the randomness of the other parties; these bounds describe the polyhedron of randomness rates (Section~\ref{sec:results}).

\begin{figure}[!t]
\centering
\begin{tikzpicture}[msg/.style={-{Stealth[length=2mm]},thick},lab/.style={font=\footnotesize,fill=white,inner sep=1pt},ttl/.style={font=\small,align=center}]
\begin{scope}
\node[party,fill=gray!20] (p1) at (0.7,1.7) {$1$};
\node[party,fill=gray!20] (p2) at (3.5,1.7) {$2$};
\node[party] (c3) at (0,0) {$3$};
\node[party] (c4) at (2.1,0) {$4$};
\node[dest]  (c5) at (4.2,0) {$5$};
\draw[msg] (p1) -- node[lab,above] {$2$ symbols} (p2);
\foreach \c in {c3,c4,c5}{\draw[msg] (p1) -- (\c); \draw[msg] (p2) -- (\c);}
\node[font=\footnotesize] at (2.1,-0.6) {$1$ symbol on each arrow to $3,4,5$};
\node[ttl] at (2.1,-1.3) {(a) Phase 1: $8$ symbols};
\end{scope}
\begin{scope}[xshift=6.6cm]
\node[party,fill=gray!20] (p1) at (0.7,1.7) {$1$};
\node[party,fill=gray!20] (p2) at (3.5,1.7) {$2$};
\node[party] (c3) at (0,0) {$3$};
\node[party] (c4) at (2.1,0) {$4$};
\node[dest]  (c5) at (4.2,0) {$5$};
\node[draw,thick,dashed,circle,minimum size=10mm] at (c3) {};
\draw[msg] (p1) -- (c3);
\draw[msg] (p2) -- (c3);
\draw[msg] (c4) -- (c3);
\draw[msg] (c3) to[bend right=28] (c5);
\node[ttl] at (2.1,-1.3) {(b) Phase 2, sum $1$ at party $3$: $4$ symbols};
\end{scope}
\begin{scope}[yshift=-4cm]
\node[party,fill=gray!20] (p1) at (0.7,1.7) {$1$};
\node[party,fill=gray!20] (p2) at (3.5,1.7) {$2$};
\node[party] (c3) at (0,0) {$3$};
\node[party] (c4) at (2.1,0) {$4$};
\node[dest]  (c5) at (4.2,0) {$5$};
\node[draw,thick,dashed,circle,minimum size=10mm] at (c4) {};
\draw[msg] (p1) -- (c4);
\draw[msg] (p2) -- (c4);
\draw[msg] (c3) -- (c4);
\draw[msg] (c4) -- (c5);
\node[ttl] at (2.1,-1.3) {(c) Phase 2, sum $2$ at party $4$: $4$ symbols};
\end{scope}
\begin{scope}[xshift=6.6cm,yshift=-4cm]
\node[party,fill=gray!20] (p1) at (0.7,1.7) {$1$};
\node[party,fill=gray!20] (p2) at (3.5,1.7) {$2$};
\node[party] (c3) at (0,0) {$3$};
\node[party] (c4) at (2.1,0) {$4$};
\node[dest]  (c5) at (4.2,0) {$5$};
\draw[msg] (p1) -- (c5);
\draw[msg] (p2) -- (c5);
\draw[msg] (c4) -- (c5);
\node[draw,thick,dashed,circle,minimum size=10.5mm] at (c5) {};
\draw[msg] (c3) to[bend right=28] (c5);
\node[ttl] at (2.1,-1.3) {(d) Phase 2, sum $3$ at party $5$: $4$ symbols};
\end{scope}
\end{tikzpicture}
\caption{A scheme for $(K,T)=(5,2)$ at extreme point $(5/3,1,0,0,0)$, with destination party $5$. (a) Phase 1, zero-sum noise: parties $1$ and $2$ (shaded) each deal an additive sharing of zero. (b)--(d) Phase 2, routing: the noisy input symbols of each sum are added up along a tree to the destination, through the party in the dashed circle. In total, $8+3\cdot4=20$ symbols are sent, $2$ on every link.}
\label{fig:scheme}
\end{figure}
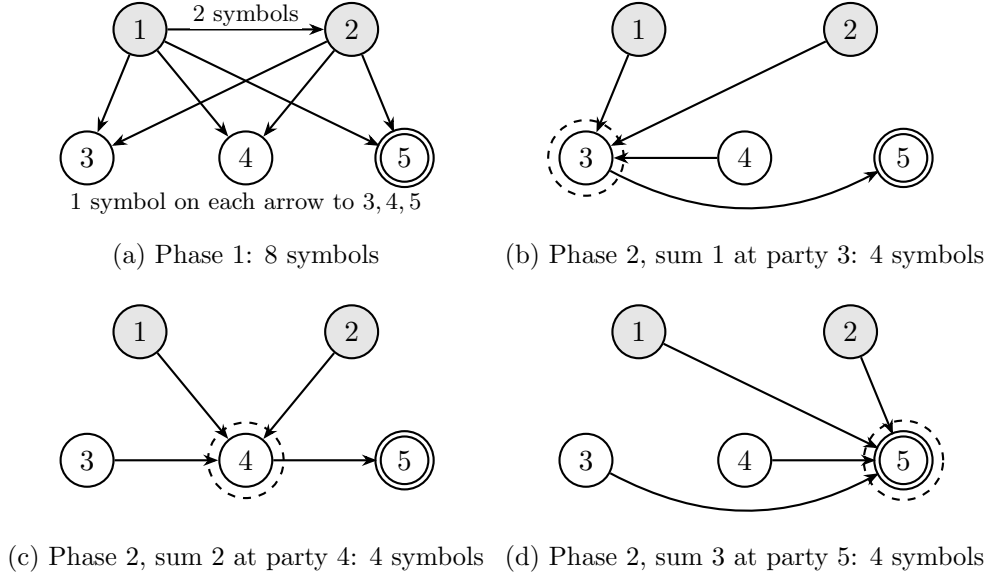

The techniques generalize in two directions (Section~\ref{sec:gen}). When every party recovers the sum, that is, for secure all-reduce~\cite{ZhouSun26}, the optimal communication rate increases by exactly $K-1$ and the polyhedron of randomness rates is unchanged. When the sum is computed at a server without input that may collude with at most $T$ of the $K$ parties, now called users, the optimal user-to-user communication rate and the polyhedron of randomness rates of the users are the same as above, the optimal server communication rate is $1$, and the server needs no randomness.

\subsection{Related Work}\label{sec:related}

In information theory, secure summation, also called secure aggregation, has been studied with one-hop communication, in which each user sends one message directly to a server that computes the sum, using correlated randomness that is independent of the inputs. The correlated randomness is supplied in advance with user dropouts~\cite{ZS22}, with groupwise structure~\cite{ZS24}, with user selection~\cite{ZS25} or with individual rates~\cite{HU26}, or generated by the users through input-independent communication~\cite{YZWT26}. Practical protocols include~\cite{Bonawitz17}, which also uses correlated randomness agreed in advance, and SwiftAgg+~\cite{Swift23}, in which users exchange shares of their inputs. These works measure correlated randomness, in total or per user~\cite{HU26}, while here every party has only independent local randomness, and its amount at every party is part of the capacity region. Data, Prabhakaran and Prabhakaran derive information-theoretic lower bounds on communication and randomness for secure computation with three parties~\cite{DPP16}, and the sum is the simplest instance of the linear secure computation problem of~\cite{ZSF19}. In theoretical computer science, the scalar results on private computation of the sum discussed above agree with our rates at $T=1$ and $T=K-2$, and the results for many instances computed together~\cite{FY92,GIS22,KOR03} are compared with ours in Appendix~\ref{app:tcs}.

\section{Problem Statement}\label{sec:model}

For a positive integer $N$, $[N]\triangleq\{1,\ldots,N\}$. For a collection of random variables indexed by a set $\cA$, write $X_{\cA}=(X_i)_{i\in\cA}$, with the indices in increasing order; vectors over $\Fq$ are column vectors, and a tuple of symbols denotes the column vector with these entries in this order.

There are $K\ge2$ parties, and party $K$ is the destination. Party $k\in[K]$ holds an input $W_k=(W_k(1),\ldots,W_k(L))$ of $L$ i.i.d.\ uniform symbols from $\Fq$ and local randomness $Z_k=(Z_k(1),\ldots,Z_k(L_{Z_k}))$ of $L_{Z_k}$ i.i.d.\ uniform symbols from $\Fq$, which party $k$ generates locally. Since the randomness is local, no randomness is initially shared among parties, and the inputs and the local randomness variables are mutually independent,
\begin{align}
&H\big((W_k)_{k\in[K]},(Z_k)_{k\in[K]}\big)=\sum_{k\in[K]}H(W_k)+\sum_{k\in[K]}H(Z_k),\label{eq:indep}\\
&H(W_k)=L,\quad H(Z_k)=L_{Z_k}~(\text{in $q$-ary units}),\quad\forall k\in[K].\label{eq:size}
\end{align}
The desired sum is $\Ws\triangleq\sum_{k\in[K]}W_k$, with symbols $W_{\Sigma}(\ell)=\sum_{k\in[K]}W_k(\ell)$, $\ell\in[L]$, so that $H(\Ws)=L$.

Every pair of parties is connected by a secure noiseless link, whose messages are observed only by its two ends. Let $\cE=\{\{i,j\}:i,j\in[K],\,i\ne j\}$ be the set of links. For $\cA\subseteq[K]$, let $\cE[\cA]$ be the links with both ends in $\cA$ and $\cE(\cA)=\cE\setminus\cE[\cA^c]$ the links with at least one end in $\cA$, and write $\cE_k=\cE(\{k\})$ for the links at party $k$. For example, in the instance of Figure~\ref{fig:model}, $\cE[\{1,2,3\}]=\{\{1,2\},\{1,3\},\{2,3\}\}$ are the thick links, $\cE(\{4,5\})$ consists of the other $7$ links, and $\cE_5=\{\{1,5\},\{2,5\},\{3,5\},\{4,5\}\}$. Communication takes place in $N$ rounds, as in~\cite{ZhouSun26}. In round $n\in[N]$, party $i$ sends to each party $j\ne i$ a message $X_{ij}(n)$ of $\Fq$ symbols, and party $j$ receives $Y_j(n)\triangleq(X_{ij}(n))_{i\in[K]\setminus\{j\}}$. The messages that a party sends in round $n$ are functions of its input, its randomness and the messages it received in earlier rounds,
\begin{equation}\label{eq:message}
H\big((X_{ij}(n))_{j\in[K]\setminus\{i\}}\,\big|\,W_i,Z_i,Y_i(1{:}n-1)\big)=0,\quad\forall i\in[K],\ n\in[N],
\end{equation}
where $Y_i(1{:}n-1)\triangleq(Y_i(1),\ldots,Y_i(n-1))$. For a link $e=\{i,j\}$, let $X_e(1{:}n)\triangleq(X_{ij}(1),X_{ji}(1),\allowbreak\ldots,\allowbreak X_{ij}(n),X_{ji}(n))$ be the messages on it in rounds $1,\ldots,n$, let $X_e\triangleq X_e(1{:}N)$ be all messages on it, and for a set $\cE'$ of links let $X_{\cE'}(1{:}n)\triangleq(X_e(1{:}n))_{e\in\cE'}$, $X_{\cE'}\triangleq X_{\cE'}(1{:}N)$ and $X=X_{\cE}$.

A protocol must satisfy two conditions.
\begin{itemize}
\item \emph{Correctness.} The destination recovers the sum from its input, its randomness and the messages it receives,
\begin{equation}\label{eq:correct}
H\big(\Ws\,\big|\,W_K,Z_K,Y_K(1{:}N)\big)=0.
\end{equation}
\item \emph{Security.} At most $T$ colluding parties learn nothing about the inputs beyond what they already know and, if they include the destination, the sum: for every $\cT\subseteq[K]$ with $|\cT|\le T$,
\begin{equation}\label{eq:security}
\begin{cases}
I\big((W_k)_{k\in[K]};(Y_k(1{:}N))_{k\in\cT}\,\big|\,\Ws,(W_k,Z_k)_{k\in\cT}\big)=0, & K\in\cT,\\[4pt]
I\big((W_k)_{k\in[K]};(Y_k(1{:}N))_{k\in\cT}\,\big|\,(W_k,Z_k)_{k\in\cT}\big)=0, & K\notin\cT.
\end{cases}
\end{equation}
\end{itemize}
We call a set $\cT$ of at most $T$ colluding parties a \emph{coalition}, the messages $(Y_k(1{:}N))_{k\in\cT}$ that it receives its \emph{observation}, and the parties outside it \emph{honest}; a set $\cH$ of $K-T$ parties is an \emph{honest set}, the honest parties of the coalition $[K]\setminus\cH$ of $T$ parties. The messages that the parties in $\cT$ send are, by~\eqref{eq:message} and the chain rule over the rounds, functions of $(W_k,Z_k)_{k\in\cT}$ and of the messages they receive, so~\eqref{eq:security} is unchanged if the received messages $(Y_k(1{:}N))_{k\in\cT}$ are replaced by $X_{\cE(\cT)}$, all messages that the parties in $\cT$ send or receive. Since $\Ws$ is a function of the inputs, the second line of~\eqref{eq:security} implies the first for $K\notin\cT$ as well, so that every protocol satisfying~\eqref{eq:security} satisfies
\begin{equation}\label{eq:security-sum}
I\big((W_k)_{k\in[K]};X_{\cE(\cT)}\,\big|\,\Ws,(W_k,Z_k)_{k\in\cT}\big)=0,\quad\forall\cT\subseteq[K],\ |\cT|\le T.
\end{equation}
The converse uses security only through~\eqref{eq:security-sum}.

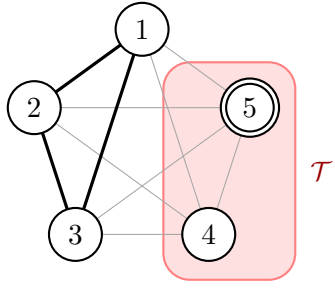
\begin{figure}[t]
\centering
\begin{tikzpicture}
\foreach \n/\a in {1/90,2/162,3/234,4/306,5/18}{\coordinate (m\n) at (\a:1.5);}
\begin{scope}[on background layer]
  \node[coal,fit=(m4)(m5),inner sep=6mm] (cm) {};
  \foreach \i/\j in {1/4,1/5,2/4,2/5,3/4,3/5,4/5}{\draw[thin link] (m\i)--(m\j);}
  \foreach \i/\j in {1/2,1/3,2/3}{\draw[hon link] (m\i)--(m\j);}
\end{scope}
\foreach \n in {1,2,3,4}{\node[party] at (m\n) {$\n$};}
\node[dest] at (m5) {$5$};
\node[red!60!black,font=\small,anchor=west] at ($(cm.east)+(0.05,0)$) {$\mathcal T$};
\end{tikzpicture}
\caption{Interactive secure summation with $(K,T)=(5,2)$; every pair of parties is connected by a secure link. Double circle: destination; shaded: a coalition $\cT=\{4,5\}$ of colluding parties; thick: the links $\cE[\{1,2,3\}]$ among the other parties, which form the honest set $\cH=\{1,2,3\}$.}
\label{fig:model}
\end{figure}

Let $L_X$ be the total length of all messages $X$ and $L_Z\triangleq\sum_{k\in[K]}L_{Z_k}$ the total length of the local randomness $(Z_k)_{k\in[K]}$. The \emph{communication rate}, the \emph{randomness rate} of party $k$ and the \emph{total randomness rate} are
\[
R_X\triangleq\frac{L_X}{L},\qquad R_{Z_k}\triangleq\frac{L_{Z_k}}{L},\quad k\in[K],\qquad R_Z\triangleq\frac{L_Z}{L}=\sum_{k\in[K]}R_{Z_k}.
\] A rate tuple $(R_X,R_{Z_1},\ldots,R_{Z_K})$ is \emph{achievable} if there is a protocol satisfying~\eqref{eq:message}, \eqref{eq:correct} and~\eqref{eq:security} with rates no greater than $R_X,R_{Z_1},\ldots,R_{Z_K}$, for some $L$. The \emph{capacity region} $\cR$ is the closure of the set of achievable rate tuples. The \emph{optimal communication rate} is the minimum of $R_X$ over $\cR$, and the \emph{optimal randomness rate} of a set of parties is the minimum over $\cR$ of the sum of their randomness rates; in particular, the optimal total randomness rate is the minimum of $R_Z$.

The parameters $K$, $T$ and $q$ are fixed. Each protocol, for a given input length $L$, has a finite number $N$ of rounds, a fixed length $L_{Z_k}$ of the local randomness of every party $k$, which is generated before the protocol starts, and messages of fixed lengths; all of them may grow with $L$, and the lengths enter only through their ratios to $L$, the rates. The capacity region is a closure over all $L$ and all protocols. The converse holds for every $N$, and it uses the length of the local randomness only through $H(Z_k)\le L_{Z_k}$, so it also holds for local randomness with any distribution on $q^{L_{Z_k}}$ values. Recovery is exact and security is perfect.\footnote{Allowing a vanishing probability of error in recovery and a vanishing normalized leakage, that is, weak security, does not change the capacity region. By Fano's inequality, the right side of~\eqref{eq:correct} becomes $o(L)$ instead of $0$, and with weak security so do the right sides of~\eqref{eq:security}; the converse then holds with additional $o(L)$ terms, which vanish after dividing by $L$.} Parties follow the protocol, and none deviates from it, for example by dropping out.

\section{Main Result}\label{sec:results}

Our main result is the characterization of the capacity region, presented in the following theorem.

\begin{theorem}\label{thm:main}
For $K\ge2$ and $0\le T\le K-2$, the capacity region $\cR$ of interactive secure summation is the set of rate tuples $(R_X,R_{Z_1},\ldots,R_{Z_K})\in\mathbb R_{\ge0}^{K+1}$ such that
\begin{equation}\label{eq:region}
\begin{aligned}
R_X&\ge\frac{K(K-1)}{K-T},\\
\sum_{k\in[K]\setminus\cA}R_{Z_k}&\ge\frac{(T-|\cA|)(K-|\cA|-1)}{K-T},\qquad\forall\,\cA\subseteq[K],\ |\cA|\le T.
\end{aligned}
\end{equation}
\end{theorem}

We make the following remarks on Theorem~\ref{thm:main}.
\begin{itemize}
\item \emph{Proof.} The converse is proved in Section~\ref{sec:converse} and holds for every protocol satisfying~\eqref{eq:message}, \eqref{eq:correct} and~\eqref{eq:security-sum}; the achievability is proved in Section~\ref{sec:scheme}.
\item \emph{Extreme points.} For $T\ge1$, the randomness constraints describe the vectors that are at least a convex combination of the extreme points
\begin{equation}\label{eq:extreme}
R_{Z_{p_t}}=\frac{K+T-2t}{K-T},\ t\in[T],\qquad R_{Z_k}=0,\ k\notin\{p_1,\ldots,p_T\},
\end{equation}
over all choices of $T$ distinct parties $p_1,\ldots,p_T$, which may include the destination; this is proved in Section~\ref{sec:extreme}. Up to a relabeling of the parties, there is a single extreme point, the vector
\[
\Big(\frac{K+T-2}{K-T},\frac{K+T-4}{K-T},\ldots,\frac{K-T}{K-T},0,\ldots,0\Big)
\]
with $T$ nonzero entries. Since the destination plays a different role, up to a relabeling of the other $K-1$ parties there are $T+1$ types of extreme points: in type $t\in[T]$ the destination is $p_t$, with randomness rate $(K+T-2t)/(K-T)$, and in type $0$ it is none of $p_1,\ldots,p_T$, with no randomness. Table~\ref{tab:extreme} lists them for $(K,T)=(5,2)$, where the three types are the three rows.
\begin{table}[!ht]
\centering
\renewcommand{\arraystretch}{1.25}
\begin{tabular}{lll}
\toprule
Destination (party $5$) & Extreme points $(R_{Z_1},\ldots,R_{Z_5})$ & Number\\
\midrule
no randomness, $p_1,p_2\in[4]$ & $(5/3,1,0,0,0)$ and permutations of the first $4$ entries & $12$\\
first, $p_1=5$ & $(1,0,0,0,5/3)$ and permutations of the first $4$ entries & $4$\\
second, $p_2=5$ & $(5/3,0,0,0,1)$ and permutations of the first $4$ entries & $4$\\
\bottomrule
\end{tabular}
\caption{The $20$ extreme points of the polyhedron of randomness rates for $(K,T)=(5,2)$, grouped by the role of the destination. Each gives randomness $5/3$ to party $p_1$ and $1$ to party $p_2$, with total $8/3$.}
\label{tab:extreme}
\end{table}

\item \emph{Simultaneous optimality.} Every extreme point is achieved together with the optimal communication rate by a linear scheme that works over every $\Fq$ (Section~\ref{sec:scheme}). Hence the optimal communication rate and every minimal allocation of the randomness are achieved simultaneously, only $T$ parties need randomness, and linear schemes are optimal. Every extreme point is achieved exactly, at a finite input length $L$, by a scheme with three rounds for $T\ge1$ and one round for $T=0$; the closure in the definition of $\cR$ is needed only for time sharing.
\item \emph{Randomness outside a set.} The constraint with a set $\cA$ of $M\le T$ parties gives the optimal randomness rate $(T-M)(K-M-1)/(K-T)$ of the parties outside $\cA$, achieved, with the optimal communication rate, by the extreme points whose $M$ largest randomness rates are at $\cA$. In particular, $\cA=\emptyset$ gives the optimal total randomness rate $T(K-1)/(K-T)$.
\item \emph{Randomness of the destination.} For $T\ge1$, while the total randomness rate is optimal, the destination can generate any amount up to $(K+T-2)/(K-T)$, but no more. A coalition that contains the destination knows its randomness, so against such coalitions the other parties face the problem with $K-1$ parties and coalitions of at most $T-1$ parties and must generate its optimal randomness rate $(T-1)(K-2)/(K-T)$, which is the constraint of~\eqref{eq:region} for $\cA=\{K\}$; the destination can generate only the rest of the total, $T(K-1)/(K-T)-(T-1)(K-2)/(K-T)=(K+T-2)/(K-T)$, and does so at the extreme points with $p_1=K$. Every amount in between is achieved by time sharing with an extreme point at which the destination has no randomness.
\item \emph{Example.} For $(K,T)=(5,2)$ the region is $R_X\ge20/3$, $\sum_{k\in[5]}R_{Z_k}\ge8/3$ and $\sum_{k\ne j}R_{Z_k}\ge1$ for every $j\in[5]$, with extreme points the permutations of $(5/3,1,0,0,0)$. The converse and the scheme for this case are worked out in Sections~\ref{sec:ex-converse} and~\ref{sec:ex-scheme}.
\item \emph{No collusion.} For $T=0$ the randomness constraints are vacuous, and $R_X\ge K-1$ is achieved without randomness by sending every input to the destination.
\item \emph{Output to every party.} When every party must recover the sum, the optimal communication rate increases by exactly $K-1$ and the randomness constraints are unchanged (Section~\ref{sec:all}).
\end{itemize}

The bounds have a simple interpretation in four steps, which the converse makes precise (Section~\ref{sec:converse}).

\emph{Step 1: one honest set.} A coalition of $T$ parties observes the messages on every link at its parties. The other $K-T$ parties form an honest set. Beyond their sum, their inputs have $(K-T-1)L$ symbols of uncertainty, and these must be carried by the messages on the $\binom{K-T}2=(K-T)(K-T-1)/2$ links inside the honest set. So these links carry at least $(K-T-1)L$ symbols, that is, $2L/(K-T)$ symbols per link on average.

\emph{Step 2: total communication.} Every set of $K-T$ parties is an honest set, and every link lies inside the same number of honest sets. Averaging Step~1 over the honest sets, the average load per link is at least $2L/(K-T)$ symbols, so the total communication is at least $\binom K2\cdot2L/(K-T)=K(K-1)L/(K-T)$.

\emph{Step 3: total randomness.} The messages and the sum together have entropy at least $L+K(K-1)L/(K-T)$, and the inputs contribute at most $KL$ of it. The randomness contributes the rest, at least $T(K-1)L/(K-T)$.

\emph{Step 4: randomness outside a set.} A coalition that contains a set of $M$ parties knows their randomness. Against such coalitions, the other $K-M$ parties face the same problem with coalitions of at most $T-M$ parties, so by Step~3 they generate at least $(T-M)(K-M-1)L/(K-T)$.

\section{Converse}\label{sec:converse}

We begin with two lemmas, then work through the converse for $(K,T)=(5,2)$ (Section~\ref{sec:ex-converse}) and prove the converse of Theorem~\ref{thm:main} (Section~\ref{sec:converse-main}).

The first lemma says that the messages on a party's links, together with the sum, determine its input. It is Fact~1 of~\cite{GIS22}, which goes back to~\cite{CK93,BDPV99}; we include a proof for completeness.

\begin{lemma}[Recoverability]\label{lem:recover}
In a protocol satisfying~\eqref{eq:message} and~\eqref{eq:correct}, for every party $k\in[K]$,\footnote{Lemma~\ref{lem:recover} does not use security at all; it holds for every protocol that computes the sum from independent inputs and randomness.}
\[
H(W_k\,|\,X_{\cE_k},\Ws)=0.
\]
\end{lemma}

\begin{proof}
The idea is that the inputs are independent, so information about $W_k$ can leave party $k$ only through the messages on its links, and the destination needs it to compute the sum. Fix $k\in[K]$. Let $\Omega_j=(W_j,Z_j)$ and let $\bar\Omega_k=(\Omega_j)_{j\ne k}$ be the data of the other parties. Recall from Section~\ref{sec:model} that $X_{\cE_k}(1{:}n)=(X_e(1{:}n))_{e\in\cE_k}$ are the messages on the links at $k$ in rounds $1,\ldots,n$, and $X_{\cE_k}(1{:}N)=X_{\cE_k}$.

\emph{Independence given the links at $k$.} We show $I(\Omega_k;\bar\Omega_k\,|\,X_{\cE_k})=0$. By~\eqref{eq:message} and induction over the rounds, the messages that $k$ sends in round $n$ are a function of $(\Omega_k,X_{\cE_k}(1{:}n-1))$, and those that it receives in round $n$ are a function of $(\bar\Omega_k,X_{\cE_k}(1{:}n-1))$, since every message on another link is a function of $\bar\Omega_k$ and of the messages that $k$ sent earlier. For random variables $Q$ and $\Psi$ with $\Psi$ a function of $(\Omega_k,Q)$, $I(\Omega_k;\bar\Omega_k\,|\,Q)=I(\Omega_k,\Psi;\bar\Omega_k\,|\,Q)\ge I(\Omega_k;\bar\Omega_k\,|\,Q,\Psi)$, and similarly for a function of $(\bar\Omega_k,Q)$. Adding the messages of each round in turn, first those that $k$ sends and then those that it receives,
\[
I(\Omega_k;\bar\Omega_k\,|\,X_{\cE_k})\le I(\Omega_k;\bar\Omega_k\,|\,X_{\cE_k}(1{:}N-1))\le\cdots\le I(\Omega_k;\bar\Omega_k)=0,
\]
where the last step is~\eqref{eq:indep}.

\emph{Recovery.} If a variable $\Phi$ is a function of $(\Omega_k,X_{\cE_k})$ and also of $(\bar\Omega_k,X_{\cE_k})$, then
\[
H(\Phi\,|\,X_{\cE_k})=I(\Phi;\Phi\,|\,X_{\cE_k})\le I(\Omega_k;\bar\Omega_k\,|\,X_{\cE_k})=0,
\]
so $\Phi$ is a function of $X_{\cE_k}$. By~\eqref{eq:correct}, $\Ws$ is a function of $(\Omega_K,Y_K(1{:}N))$, the data and the observation of the destination. If $k\ne K$, then $(\Omega_K,Y_K(1{:}N))$ is a function of $(\bar\Omega_k,X_{\cE_k})$, since $\Omega_k$ affects the other parties only through $X_{\cE_k}$, so $\Phi=W_k=\Ws-\sum_{j\ne k}W_j$ qualifies, and $H(W_k\,|\,X_{\cE_k})=0$. If $k=K$, then $(\Omega_K,Y_K(1{:}N))$ is a function of $(\Omega_K,X_{\cE_K})$, so $\Phi=\sum_{j\ne K}W_j=\Ws-W_K$ qualifies, and $W_K=\Ws-\Phi$ gives $H(W_K\,|\,X_{\cE_K},\Ws)=0$.
\end{proof}

The second lemma is the conditional form of Shearer's lemma~\cite{CGFS86} together with its form for complements.

\begin{lemma}[Conditional Shearer]\label{lem:cover}
Let $(X_e)_{e\in\cI}$ and $Q$ be arbitrary discrete random variables, and let $\cF$ be a collection of subsets of $\cI$, possibly with repetitions, such that every $e\in\cI$ lies in exactly $D$ sets of $\cF$. Then
\begin{align}
\sum_{\mathcal G\in\cF}H(X_{\mathcal G}\,|\,Q)&\ge D\,H(X_{\cI}\,|\,Q),\label{eq:shearer}\\
\sum_{\mathcal G\in\cF}H(X_{\mathcal G}\,|\,X_{\cI\setminus\mathcal G},Q)&\le D\,H(X_{\cI}\,|\,Q).\label{eq:shearer-dual}
\end{align}
\end{lemma}

\begin{proof}
Fix an order on $\cI$ and write $X_{<e}$ for the variables $X_{e'}$ with $e'$ before $e$.

\emph{Inequality~\eqref{eq:shearer}.} For each $\mathcal G\in\cF$,
\[
H(X_{\mathcal G}\,|\,Q)\overset{(a)}{=}\sum_{e\in\mathcal G}H\big(X_e\,\big|\,X_{\mathcal G\cap\{e':e'<e\}},Q\big)\overset{(b)}{\ge}\sum_{e\in\mathcal G}H(X_e\,|\,X_{<e},Q),
\]
where (a) is the chain rule over the elements of $\mathcal G$ in order, and (b) holds because conditioning reduces entropy and $\mathcal G\cap\{e':e'<e\}$ is contained in $\{e':e'<e\}$. Summing over $\mathcal G\in\cF$ and exchanging the order of summation,
\begin{align*}
\sum_{\mathcal G\in\cF}\sum_{e\in\mathcal G}H(X_e\,|\,X_{<e},Q)&=\sum_{e\in\cI}\big|\{\mathcal G\in\cF:e\in\mathcal G\}\big|\,H(X_e\,|\,X_{<e},Q)\\
&=D\sum_{e\in\cI}H(X_e\,|\,X_{<e},Q)=D\,H(X_{\cI}\,|\,Q),
\end{align*}
where the last step is the chain rule.

\emph{Inequality~\eqref{eq:shearer-dual}.} Every $e\in\cI$ lies in exactly $|\cF|-D$ of the complements $\cI\setminus\mathcal G$, $\mathcal G\in\cF$, so~\eqref{eq:shearer} applied to the complements gives $\sum_{\mathcal G\in\cF}H(X_{\cI\setminus\mathcal G}\,|\,Q)\ge(|\cF|-D)H(X_{\cI}\,|\,Q)$. Since $H(X_{\mathcal G}\,|\,X_{\cI\setminus\mathcal G},Q)=H(X_{\cI}\,|\,Q)-H(X_{\cI\setminus\mathcal G}\,|\,Q)$, summing over $\mathcal G\in\cF$ gives~\eqref{eq:shearer-dual}.
\end{proof}

\subsection{Example: The Converse for $(K,T)=(5,2)$}\label{sec:ex-converse}

We first work through the case $(K,T)=(5,2)$ of Figure~\ref{fig:model}, the smallest case strictly between $T=1$ and $T=K-2$, in which the randomness outside a single party is constrained and the honest sets overlap in more than one party.

\begin{example}\label{ex:converse}
Let $(K,T)=(5,2)$ and fix a protocol satisfying~\eqref{eq:message}--\eqref{eq:security}. We use the notation of Section~\ref{sec:model}; for example, $X_{\cE[\{1,2,3\}]}$ are the messages on the $3$ links $\{1,2\},\{1,3\},\{2,3\}$ inside $\{1,2,3\}$, and $X_{\cE(\{4,5\})}$ those on the other $7$ links $\{1,4\},\{2,4\},\{3,4\},\{1,5\},\{2,5\},\{3,5\},\{4,5\}$, the links at $4$ or $5$.

\emph{Recoverability.} By Lemma~\ref{lem:recover}, $W_1$ is a function of $\Ws$ and $X_{\cE_1}$, and similarly for every party.

\emph{Coalition bound.} Take the honest set $\{1,2,3\}$ and the coalition $\{4,5\}$, which sends or receives $X_{\cE(\{4,5\})}$. Then
\begin{align*}
H(X_{\cE[\{1,2,3\}]}\,|\,X_{\cE(\{4,5\})},\Ws)
&\ge H(X_{\cE[\{1,2,3\}]}\,|\,X_{\cE(\{4,5\})},\Ws,W_{\{4,5\}},Z_{\{4,5\}})\\
&\ge I(W_{\{1,2,3\}};X_{\cE[\{1,2,3\}]}\,|\,X_{\cE(\{4,5\})},\Ws,W_{\{4,5\}},Z_{\{4,5\}})\\
&\overset{(a)}{=}H(W_{\{1,2,3\}}\,|\,X_{\cE(\{4,5\})},\Ws,W_{\{4,5\}},Z_{\{4,5\}})\\
&\overset{(b)}{=}H(W_{\{1,2,3\}}\,|\,\Ws,W_{\{4,5\}},Z_{\{4,5\}})\\
&\overset{(c)}{=}H(W_1,W_2,W_3\,|\,W_1+W_2+W_3)=2L,
\end{align*}
where (a) holds because, by Recoverability, each of $W_1,W_2,W_3$ is a function of $\Ws$ and of $X=(X_{\cE[\{1,2,3\}]},X_{\cE(\{4,5\})})$; (b) is~\eqref{eq:security-sum} for the coalition $\{4,5\}$; and (c) holds because $W_1,W_2,W_3$ are independent of $(W_{\{4,5\}},Z_{\{4,5\}})$, and $\Ws$ together with $W_{\{4,5\}}$ is equivalent to $W_1+W_2+W_3$ together with $W_{\{4,5\}}$. In words, the inputs of the honest parties, beyond their sum, are hidden from the coalition, so their $2L$ symbols of uncertainty must be carried by the $3$ links inside the honest set.

\emph{Covering.} The same bound holds for each of the $10$ honest sets $\cH\subseteq[5]$ with $|\cH|=3$,
\[
H(X_{\cE[\cH]}\,|\,X_{\cE([5]\setminus\cH)},\Ws)\ge2L.
\]
Each link lies inside exactly $3$ of the $10$ honest sets; for example, the link $\{1,2\}$ lies inside $\{1,2,3\}$, $\{1,2,4\}$ and $\{1,2,5\}$. Hence, by~\eqref{eq:shearer-dual} of Lemma~\ref{lem:cover} with $\cI=\cE$, $Q=\Ws$ and $D=3$, the sum of the $10$ left sides is at most $3H(X\,|\,\Ws)$, and
\[
3H(X\,|\,\Ws)\ge10\cdot2L,\qquad\text{so}\qquad H(X\,|\,\Ws)\ge\frac{20}{3}L.
\]
Equivalently, the average load per link is at least $2L/3$ symbols.

\emph{Rate.} Since $X$ has $L_X$ symbols, $L_X\ge H(X)\ge H(X\,|\,\Ws)\ge 20L/3$ and $R_X\ge20/3$. For randomness, the messages and the sum are functions of all inputs and randomness by~\eqref{eq:message}, so by~\eqref{eq:indep} and~\eqref{eq:size},
\[
5L+L_Z=H(W_{[5]},Z_{[5]})\ge H(X,\Ws)=H(\Ws)+H(X\,|\,\Ws)\ge L+\frac{20}{3}L,
\]
so $L_Z\ge 8L/3$ and $R_Z\ge8/3$.

\emph{Randomness outside a set of parties.} Consider now the set $\cA=\{5\}$ and only the $4$ honest sets $\cH\subseteq[4]$, whose coalitions contain party $5$, and condition also on $Q=(X_{\cE_5},\Ws,W_5,Z_5)$, which these coalitions know. The coalition bound gives $H(X_{\cE[\cH]}\,|\,X_{\cE([5]\setminus\cH)},\Ws,W_5,Z_5)\ge2L$, and each of the $6$ links inside $[4]$ lies inside $2$ of these honest sets, so by~\eqref{eq:shearer-dual} of Lemma~\ref{lem:cover}, $2H(X_{\cE[\{1,2,3,4\}]}\,|\,Q)\ge4\cdot2L$. Hence
\begin{align*}
4L+\sum_{k\in[4]}L_{Z_k}&=H(W_{[4]},Z_{[4]}\,|\,W_5,Z_5)\ge H(X,\Ws\,|\,W_5,Z_5)\\
&=H(X_{\cE_5},\Ws\,|\,W_5,Z_5)+H(X_{\cE[\{1,2,3,4\}]}\,|\,Q)\ge L+4L,
\end{align*}
so $R_{Z_1}+R_{Z_2}+R_{Z_3}+R_{Z_4}\ge1$. The same steps with any party $j$ in place of party $5$ give $\sum_{k\ne j}R_{Z_k}\ge1$; for sets $\cA$ of $2$ parties the bound is $0$.
\end{example}

\subsection{Proof of the Converse of Theorem~\ref{thm:main}}\label{sec:converse-main}

The general case follows the same steps, with honest sets of size $K-T\ge2$ in place of size $3$. Fix a protocol satisfying~\eqref{eq:message}--\eqref{eq:security}.

\emph{Recoverability.} By Lemma~\ref{lem:recover}, $H(W_k\,|\,X_{\cE_k},\Ws)=0$ for every $k\in[K]$.

\emph{Coalition bound.} Let $\cH\subseteq[K]$ with $|\cH|=K-T$ and $\cT=\cH^c$, so $|\cT|=T$ and $\cE=\cE[\cH]\cup\cE(\cT)$. As in Example~\ref{ex:converse},
\[
\begin{aligned}
H(X_{\cE[\cH]}\,|\,X_{\cE(\cT)},\Ws,W_{\cT},Z_{\cT})&\ge H(W_{\cH}\,|\,X_{\cE(\cT)},\Ws,W_{\cT},Z_{\cT})\\
&=H(W_{\cH}\,|\,\Ws,W_{\cT},Z_{\cT})=(K-T-1)L,
\end{aligned}
\]
where the inequality uses Recoverability for every $k\in\cH$, the first equality is~\eqref{eq:security-sum}, and the second equality holds because $K-T$ independent uniform vectors given their sum have entropy $(K-T-1)L$. Removing $(W_{\cT},Z_{\cT})$ from the conditioning gives $H(X_{\cE[\cH]}\,|\,X_{\cE(\cT)},\Ws)\ge(K-T-1)L$. For $T=0$, $\cT$ is empty and the bound reads $H(X\,|\,\Ws)\ge(K-1)L$.

\emph{Covering.} Apply~\eqref{eq:shearer-dual} of Lemma~\ref{lem:cover} with $\cI=\cE$, $Q=\Ws$ and $\cF=\{\cE[\cH]:|\cH|=K-T\}$. A link $\{i,j\}$ lies in $\cE[\cH]$ if and only if $i,j\in\cH$, which holds for $\binom{K-2}{K-T-2}$ of the $\binom{K}{K-T}$ honest sets $\cH$. Hence $D=\binom{K-2}{K-T-2}$, and summing the coalition bound over all honest sets gives
\[
\binom{K}{K-T}(K-T-1)L\le\binom{K-2}{K-T-2}H(X\,|\,\Ws),
\qquad\text{so}\qquad
H(X\,|\,\Ws)\ge\frac{K(K-1)}{K-T}L.
\]

\emph{Rate.} As in Example~\ref{ex:converse}, $L_X\ge H(X\,|\,\Ws)\ge K(K-1)L/(K-T)$, and
\[
KL+L_Z=H(W_{[K]},Z_{[K]})\ge H(X,\Ws)=H(\Ws)+H(X\,|\,\Ws)\ge L+\frac{K(K-1)}{K-T}L,
\]
so $L_Z\ge(K(K-1)/(K-T)-(K-1))L=T(K-1)L/(K-T)$.

\emph{Randomness outside a set of parties.} Let $\cA\subseteq[K]$ with $|\cA|=M$, $1\le M\le T$; the case $\cA=\emptyset$ is the bound on $R_Z$ above. As in Example~\ref{ex:converse}, apply the coalition bound only to honest sets $\cH\subseteq[K]\setminus\cA$ with $|\cH|=K-T$, whose coalitions $\cT=[K]\setminus\cH$ contain $\cA$, and let $Q=(X_{\cE(\cA)},\Ws,W_{\cA},Z_{\cA})$. The links of $\cE[\cA^c]$ outside $\cE[\cH]$ and the links of $\cE(\cA)$ lie in $\cE(\cT)$, so the coalition bound gives
\[
H\big(X_{\cE[\cH]}\,\big|\,X_{\cE[\cA^c]\setminus\cE[\cH]},Q\big)\ge(K-T-1)L.
\]
Each link of $\cE[\cA^c]$ lies in $\binom{K-M-2}{K-T-2}$ of the $\binom{K-M}{K-T}$ sets $\cH$, so~\eqref{eq:shearer-dual} of Lemma~\ref{lem:cover} with $\cI=\cE[\cA^c]$ gives $H(X_{\cE[\cA^c]}\,|\,Q)\ge(K-M)(K-M-1)L/(K-T)$, and
\begin{align*}
(K-M)L+\sum_{k\notin\cA}L_{Z_k}&=H\big((W_k,Z_k)_{k\notin\cA}\,\big|\,W_{\cA},Z_{\cA}\big)\ge H(X,\Ws\,|\,W_{\cA},Z_{\cA})\\
&=H(X_{\cE(\cA)},\Ws\,|\,W_{\cA},Z_{\cA})+H(X_{\cE[\cA^c]}\,|\,Q)\\
&\ge L+\frac{(K-M)(K-M-1)}{K-T}L,
\end{align*}
where the first inequality holds because $X$ and $\Ws$ are functions of all inputs and randomness by~\eqref{eq:message}, and the last uses $H(\Ws\,|\,W_{\cA},Z_{\cA})=L$, since $M<K$. As $(K-M)(K-M-1)/(K-T)+1-(K-M)=(T-M)(K-M-1)/(K-T)$, this is the constraint of~\eqref{eq:region} for $\cA$.\qed

\section{Achievability}\label{sec:scheme}

This section reduces achievability to the extreme points (Section~\ref{sec:extreme}), works through a scheme at an extreme point for $(K,T)=(5,2)$ (Section~\ref{sec:ex-scheme}), and gives a scheme at every extreme point (Section~\ref{sec:scheme-main}). For $T=0$, sending every input to the destination achieves $R_X=K-1$ without randomness, so let $T\ge1$; recall that $K-T\ge2$.

\subsection{Extreme Points}\label{sec:extreme}

The randomness constraints of~\eqref{eq:region} depend on $\cA$ only through its size. Let $g(m)=0$ for $0\le m\le K-T$ and $g(m)=(m-K+T)(m-1)/(K-T)$ for $K-T<m\le K$. For $|\cA|=M\le T$, the right side of~\eqref{eq:region} is $g(K-M)$, and for a set of at most $K-T$ parties $g$ is $0$, so the randomness part of $\cR$ is
\[
\mathcal P\triangleq\Big\{\boldsymbol\mu\in\mathbb R_{\ge0}^K:\ \sum_{k\in\cB}\mu_k\ge g(|\cB|)\ \text{for all }\cB\subseteq[K]\Big\}.
\]
As a polyhedron in $\mathbb R_{\ge0}^K$ that contains every vector that is at least one of its points in every entry, $\mathcal P$ is the set of vectors that are at least a convex combination of its extreme points. The next lemma shows that these are the points~\eqref{eq:extreme}, a known property of such polyhedra~\cite{Edmonds70}; we include a short proof.

\begin{lemma}[Extreme Points]\label{lem:extreme}
For $T\ge1$, the extreme points of $\mathcal P$ are the points~\eqref{eq:extreme}.
\end{lemma}

\begin{proof}
Let $d(m)\triangleq g(m)-g(m-1)$, which is $0$ for $m\le K-T$ and $(2m-2-K+T)/(K-T)$ for $m>K-T$. Then $d(1)\le\cdots\le d(K)$, $g(m)=d(1)+\cdots+d(m)$, and, since $d(K-t+1)=(K+T-2t)/(K-T)$, $t\in[T]$, the points~\eqref{eq:extreme} are the permutations of $\boldsymbol d=(d(1),\ldots,d(K))$. Write $\mu(\cB)=\sum_{k\in\cB}\mu_k$, and call $\cB$ tight at $\boldsymbol\mu\in\mathcal P$ if $\mu(\cB)=g(|\cB|)$; since $g(1)=0$, the constraint $\mu_k\ge0$ is the constraint for $\cB=\{k\}$.

\emph{The points~\eqref{eq:extreme} are extreme.} Since $\boldsymbol d$ is nondecreasing, any $m$ of its entries add up to at least $d(1)+\cdots+d(m)=g(m)$. Hence $\boldsymbol d\in\mathcal P$, and the $K$ sets $[m]$, $m\in[K]$, are tight at $\boldsymbol d$. Their constraints are linearly independent; $\boldsymbol d$ is therefore an extreme point, and by symmetry every permutation of $\boldsymbol d$ is one as well.

\emph{There are no others.} Let $\boldsymbol\mu$ be an extreme point, at which $K$ linearly independent constraints are tight. The set function $\cB\mapsto g(|\cB|)$ is supermodular:
\[
g(|\cB\cup\cB'|)-g(|\cB'|)=\sum_{m=|\cB'|+1}^{|\cB\cup\cB'|}d(m)\ge\sum_{m=|\cB\cap\cB'|+1}^{|\cB|}d(m)=g(|\cB|)-g(|\cB\cap\cB'|),
\]
since both sums have $|\cB|-|\cB\cap\cB'|$ terms, $|\cB'|\ge|\cB\cap\cB'|$ and $\boldsymbol d$ is nondecreasing. If $\cB$ and $\cB'$ are tight, then
\[
\mu(\cB\cup\cB')+\mu(\cB\cap\cB')=\mu(\cB)+\mu(\cB')=g(|\cB|)+g(|\cB'|)\le g(|\cB\cup\cB'|)+g(|\cB\cap\cB'|),
\]
and as $\mu(\cB\cup\cB')\ge g(|\cB\cup\cB'|)$ and $\mu(\cB\cap\cB')\ge g(|\cB\cap\cB'|)$, both $\cB\cup\cB'$ and $\cB\cap\cB'$ are tight. Take a maximal chain of tight sets $\emptyset=\mathcal S_0\subset\mathcal S_1\subset\cdots\subset\mathcal S_J$. Every tight $\cB$ is a union of blocks $\mathcal S_j\setminus\mathcal S_{j-1}$. Indeed, for $j\in[J]$ the set $(\cB\cup\mathcal S_{j-1})\cap\mathcal S_j$ is tight and lies between $\mathcal S_{j-1}$ and $\mathcal S_j$; by maximality of the chain it equals one of them, that is, $\cB$ contains the block $\mathcal S_j\setminus\mathcal S_{j-1}$ entirely or not at all. Likewise $\cB\cup\mathcal S_J$ is tight and contains $\mathcal S_J$, hence equals $\mathcal S_J$, and $\cB\subseteq\mathcal S_J$. With $\mu(\mathcal S_0)=0$,
\[
\mu(\cB)=\sum_{j\in[J]:\ \mathcal S_j\setminus\mathcal S_{j-1}\subseteq\cB}\big(\mu(\mathcal S_j)-\mu(\mathcal S_{j-1})\big)
\]
for every $\boldsymbol\mu$, so the left side of the constraint of $\cB$ is a linear combination, with coefficients $0$ and $\pm1$, of the left sides $\mu(\mathcal S_1),\ldots,\mu(\mathcal S_J)$ of the constraints of the chain. The $K$ linearly independent tight constraints therefore lie in the span of these $J$ constraints, so $J\ge K$; a chain of distinct nonempty subsets of $[K]$ has $J\le K$, and hence $J=K$ and $|\mathcal S_j|=j$. The entry of $\boldsymbol\mu$ at the party in $\mathcal S_j\setminus\mathcal S_{j-1}$ is then $g(j)-g(j-1)=d(j)$, and $\boldsymbol\mu$ is a permutation of $\boldsymbol d$.
\end{proof}

Rates are upper bounds, and time sharing between schemes with the same communication rate keeps that rate, so by Lemma~\ref{lem:extreme} it suffices to achieve each extreme point~\eqref{eq:extreme} together with $R_X=K(K-1)/(K-T)$. The scheme of Section~\ref{sec:scheme-main} works for any choice of the parties $p_1,\ldots,p_T$ that hold randomness and of the destination, because the role of the destination enters only phase 2. Hence it achieves every extreme point, of each of the $T+1$ types: in type $t\in[T]$ the destination is the party $p_t$ with randomness rate $(K+T-2t)/(K-T)$, and in type $0$ the destination has no randomness. Example~\ref{ex:scheme} is of type $0$ for $(K,T)=(5,2)$.

\subsection{Example: An Extreme Point for $(K,T)=(5,2)$}\label{sec:ex-scheme}

We first work through the scheme of Figure~\ref{fig:scheme}, for the case $(K,T)=(5,2)$ of Example~\ref{ex:converse} and the extreme point $(5/3,1,0,0,0)$, at which the destination has no randomness. As described in Section~\ref{sec:intro}, the scheme has two phases: phase 1 prepares zero-sum noise for each of the three sums, with the randomness of parties $p_1=1$ and $p_2=2$, the \emph{suppliers}, and without the inputs, and phase 2 routes the three noisy sums, sum $\ell$ through party $c_\ell=\ell+2$; parties $3,4,5$ are the \emph{collectors}. Security rests on both phases: a coalition sees the noisy input symbols of the other parties only in the sums collected by its members, and in these sums the noise is a full-rank function of the randomness that the coalition does not see; a coalition that contains the destination also receives partial sums, which are determined by the sum and its own noisy input.

\begin{example}\label{ex:scheme}
Let $(K,T)=(5,2)$ and $L=3$, and let $\boldsymbol\Gamma$ be a $3\times2$ Cauchy matrix over $\Fq$, $q\ge5$, whose square submatrices are all invertible; in particular, for $\mathcal J\subseteq[3]$ with $|\mathcal J|=2$, the $2\times2$ submatrix $\boldsymbol\Gamma(\mathcal J,[2])$ with rows in $\mathcal J$ and columns in $[2]$ is invertible; in general, $\boldsymbol\Gamma(\mathcal J,\cU)$ denotes the submatrix of $\boldsymbol\Gamma$ with rows in $\mathcal J$ and columns in $\cU$, in increasing order.

\emph{Phase 1: zero-sum noise.} Each supplier $t\in[2]$ deals an additive sharing of zero: the shares $s_{t,k}$ to the parties $k\in[5]$, with $\sum_{k\in[5]}s_{t,k}=0$, row $t$ of Figure~\ref{fig:shares}. Party $1$ draws $s_{1,2}$ and $s_{2,1}$, the two shares between parties $1$ and $2$, and $s_{1,3},s_{1,4},s_{1,5}$; party $2$ draws $s_{2,3},s_{2,4},s_{2,5}$; all are independent and uniform, and parties $3,4,5$ have no randomness. The remaining shares are
\[
s_{1,1}=-(s_{1,2}+s_{1,3}+s_{1,4}+s_{1,5}),\qquad s_{2,2}=-(s_{2,1}+s_{2,3}+s_{2,4}+s_{2,5}).
\]
Drawing $s_{2,1}$ at party $1$ reduces the randomness of party $2$. The \emph{share vector} of party $k$ is column $k$ of Figure~\ref{fig:shares}, $\boldsymbol s_k=(s_{1,k},s_{2,k})$, and its \emph{noise} is $S_k=\boldsymbol\Gamma\boldsymbol s_k$, whose $\ell$-th symbol $S_k(\ell)$ is its noise for sum $\ell\in[3]$. The noise adds up to zero because the shares of each supplier do. In round 1, party $1$ sends $s_{1,2},s_{2,1}$ to party $2$, and parties $1$ and $2$ send $s_{1,k}$ and $s_{2,k}$ to party $k\in\{3,4,5\}$ (Figure~\ref{fig:scheme}(a)). Then every party knows its share vector and hence its noise.
\begin{figure}[t]
\centering
\begin{tikzpicture}[x=1.3cm,y=0.7cm,font=\small]
\foreach \k in {2,3,4,5} \fill[blue!18] (\k-1,0) rectangle (\k,-1);
\fill[blue!18] (0,-1) rectangle (1,-2);
\foreach \k in {3,4,5} \fill[orange!30] (\k-1,-1) rectangle (\k,-2);
\fill[gray!35] (0,0) rectangle (1,-1);
\fill[gray!35] (1,-1) rectangle (2,-2);
\draw (0,0) grid[xstep=1,ystep=1] (5,-2);
\draw[thick] (0,0) rectangle (5,-2) (2,0)--(2,-2);
\foreach \t in {1,2} \foreach \k in {1,...,5} \node at (\k-0.5,-\t+0.5) {$s_{\t,\k}$};
\foreach \k in {1,...,5} {\node[above] at (\k-0.5,0) {$k=\k$}; \node[below] at (\k-0.5,-2) {$\boldsymbol s_{\k}$};}
\node[anchor=base] at (1,0.85) {suppliers}; \node[anchor=base] at (3.5,0.85) {collectors};
\node[left] at (0,-0.5) {$t=1$}; \node[left] at (0,-1.5) {$t=2$};
\node[right] at (5,-0.5) {row sum $=0$}; \node[right] at (5,-1.5) {row sum $=0$};
\begin{scope}[shift={(0,-3.2)}]
\fill[blue!18] (0,0) rectangle (0.27,-0.43); \node[right] at (0.25,-0.215) {drawn by party $1$};
\fill[orange!30] (2.6,0) rectangle (2.87,-0.43); \node[right] at (2.85,-0.215) {drawn by party $2$};
\fill[gray!35] (5.2,0) rectangle (5.47,-0.43); \node[right] at (5.45,-0.215) {$s_{t,t}=-\sum_{k\ne t}s_{t,k}$};
\end{scope}
\end{tikzpicture}
\caption{The shares in Example~\ref{ex:scheme}: row $t$ is the additive sharing of zero dealt by supplier $t$, and column $k$ is the share vector $\boldsymbol s_k$. In general, supplier $p_t$ draws its shares to the collectors and both shares between $p_t$ and $p_{t'}$, $t<t'$.}
\label{fig:shares}
\end{figure}

\emph{Phase 2: routing.} Let $\tilde W_k(\ell)=W_k(\ell)+S_k(\ell)$ be the \emph{noisy input} symbols (Figure~\ref{fig:scheme}(b)--(d)). In round 2, for every $\ell\in[3]$, every party $k\notin\{5,\ell+2\}$ sends $\tilde W_k(\ell)$ to party $\ell+2$. In round 3, party $\ell+2$, $\ell\in[2]$, sends the \emph{partial sum} $P_\ell=\sum_{k\in[4]}\tilde W_k(\ell)$ to party $5$.

\emph{Correctness.} Party $5$ outputs $W_{\Sigma}(\ell)=P_\ell+\tilde W_5(\ell)$, $\ell\in[2]$, and $W_{\Sigma}(3)=\sum_{k\in[5]}\tilde W_k(3)$. Since the noise adds up to zero over the parties, $\sum_{k\in[5]}\tilde W_k(\ell)=W_\Sigma(\ell)$ for every $\ell\in[3]$.

\emph{Security.} Consider first the coalition $\cT=\{4,5\}$, which contains the destination. The proof uses the following sets (Figure~\ref{fig:rank}). The coalition collects the sums $\mathcal J=\{2,3\}$, so it sees noisy input symbols only in these sums, the rows $\mathcal J$ of $\boldsymbol\Gamma$. It misses the suppliers $\cU=\{1,2\}$, so it does not see their shares to the honest parties, the columns $\cU$ of $\boldsymbol\Gamma$. Of the honest parties $\cH=\{1,2,3\}$, we omit $k_0=2$, whose noisy input symbols follow from the sum, and keep $\cH'=\{1,3\}$. Let $Y_{\cT}=(Y_4(1{:}3),Y_5(1{:}3))$ be the messages that the coalition receives and $V_{\cT}=(W_4,W_5,\Ws)$ the conditioning in~\eqref{eq:security}, as $Z_4,Z_5$ are empty. By the definition of mutual information,
\[
I\big(W_{[5]};Y_{\cT}\,\big|\,V_{\cT}\big)=H\big(Y_{\cT}\,\big|\,V_{\cT}\big)-H\big(Y_{\cT}\,\big|\,W_{[5]}\big).
\]

\begin{table}[t]
\centering\footnotesize
\renewcommand{\arraystretch}{1.2}
\begin{tabular}{lll}
\toprule
Round & Example~\ref{ex:scheme}, $\cT=\{4,5\}$ & General\\
\midrule
1 & $s_{1,4},s_{1,5},s_{2,4},s_{2,5}$ & (i) shares from honest suppliers\\
2 & $\tilde W_1(2),\tilde W_1(3),\tilde W_3(2),\tilde W_3(3)$ & (ii) $\tilde W_k(\ell)$, $k\in\cH'$, $\ell\in\mathcal J$\\
2 & $\tilde W_2(\ell)=\Ws(\ell)-\sum_{k\ne2}\tilde W_k(\ell)$, $\ell=2,3$ & $\tilde W_{k_0}(\ell)=\Ws(\ell)-\sum_{k\ne k_0}\tilde W_k(\ell)$, $\ell\in\mathcal J$, if $K\in\cT$\\
3 & $P_1=\Ws(1)-\tilde W_5(1)$ & $P_\ell=\Ws(\ell)-\tilde W_K(\ell)$, $\ell\notin\mathcal J$, if $K\in\cT$\\
1--3 & $\tilde W_4(3),P_2$, known & symbols sent by parties in $\cT$, known\\
\bottomrule
\end{tabular}
\caption{The observation $Y_{\cT}$ of a coalition $\cT$: (i) and (ii) form $Y'_{\cT}$, and the other symbols are functions of $V_{\cT}$ and $Y'_{\cT}$. If $K\notin\cT$, then $k_0=K$ sends no noisy input symbols or partial sums.}\label{tab:observation}
\end{table}

\emph{Counting.} Table~\ref{tab:observation} lists the symbols that the coalition receives. Parties $4,5$ know their share vectors from the shares they receive, hence their noisy input symbols given $W_4,W_5$, so, given $V_{\cT}$, all these symbols are functions of the $8$ symbols
\[
Y'_{\cT}=\big(s_{1,4},s_{1,5},s_{2,4},s_{2,5},\tilde W_1(2),\tilde W_1(3),\tilde W_3(2),\tilde W_3(3)\big).
\]
Hence the first term is at most $8$.

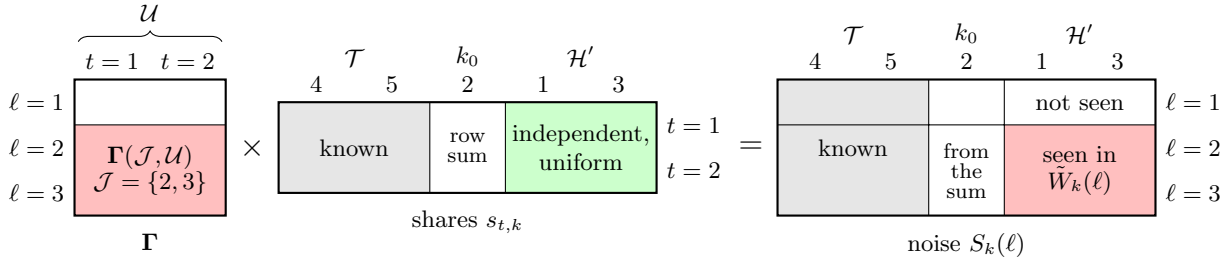
\begin{figure}[t]
\centering
\begin{tikzpicture}[x=1cm,y=1cm,font=\footnotesize]
\fill[red!25] (0,-0.6) rectangle (2,-1.8);
\draw[thick] (0,0) rectangle (2,-1.8);
\draw (0,-0.6)--(2,-0.6);
\node[align=center] at (1,-1.2) {$\boldsymbol\Gamma(\mathcal J,\cU)$\\[-1pt]$\mathcal J=\{2,3\}$};
\node[above] at (0.5,0) {$t=1$};\node[above] at (1.5,0) {$t=2$};
\foreach \l in {1,2,3} \node[left] at (0,-0.6*\l+0.3) {$\ell=\l$};
\node[below] at (1,-1.9) {$\boldsymbol\Gamma$};
\draw[decorate,decoration={brace,amplitude=4pt}] (0.05,0.5)--(1.95,0.5) node[midway,above=4pt] {$\cU$};
\node at (2.35,-0.9) {\large$\times$};
\begin{scope}[shift={(2.7,-0.3)}]
\fill[gray!20] (0,0) rectangle (2,-1.2);
\fill[green!20] (3,0) rectangle (5,-1.2);
\draw[thick] (0,0) rectangle (5,-1.2);
\draw (2,0)--(2,-1.2) (3,0)--(3,-1.2);
\node at (1,-0.6) {known};
\node[align=center,font=\scriptsize] at (2.5,-0.6) {row\\[-2pt]sum};
\node[align=center] at (4,-0.6) {independent,\\ uniform};
\foreach \k/\x in {4/0.5,5/1.5,2/2.5,1/3.5,3/4.5} \node[above] at (\x,0) {$\k$};
\node[above] at (1,0.35) {$\cT$};\node[above] at (2.5,0.35) {$k_0$};\node[above] at (4,0.35) {$\cH'$};
\node[right] at (5,-0.3) {$t=1$};\node[right] at (5,-0.9) {$t=2$};
\node[below] at (2.5,-1.3) {shares $s_{t,k}$};
\end{scope}
\node at (8.95,-0.9) {\large$=$};
\begin{scope}[shift={(9.3,0)}]
\fill[gray!20] (0,0) rectangle (2,-1.8);
\fill[red!25] (3,-0.6) rectangle (5,-1.8);
\draw[thick] (0,0) rectangle (5,-1.8);
\draw (0,-0.6)--(5,-0.6) (2,0)--(2,-1.8) (3,0)--(3,-1.8);
\node at (1,-0.9) {known};
\node at (4,-0.3) {not seen};
\node[align=center] at (4,-1.2) {seen in\\[-1pt]$\tilde W_k(\ell)$};
\node[align=center,font=\scriptsize] at (2.5,-1.2) {from\\[-2pt]the\\[-2pt]sum};
\foreach \k/\x in {4/0.5,5/1.5,2/2.5,1/3.5,3/4.5} \node[above] at (\x,0) {$\k$};
\node[above] at (1,0.35) {$\cT$};\node[above] at (2.5,0.35) {$k_0$};\node[above] at (4,0.35) {$\cH'$};
\foreach \l in {1,2,3} \node[right] at (5,-0.6*\l+0.3) {$\ell=\l$};
\node[below] at (2.5,-1.9) {noise $S_k(\ell)$};
\end{scope}
\end{tikzpicture}
\caption{What the coalition $\cT=\{4,5\}$ of Example~\ref{ex:scheme} sees, given the inputs, with the columns reordered. It sees the noisy input symbols $\tilde W_k(\ell)=W_k(\ell)+S_k(\ell)$ only in the red block: the rows $\mathcal J=\{2,3\}$, the sums it collects, and the columns $\cH'=\{1,3\}$. Their noise is $\boldsymbol\Gamma(\mathcal J,\cU)$, with $\cU=\{1,2\}$ indexing the suppliers that $\cT$ misses, times the shares in the green block, which $\cT$ does not see and which are independent and uniform; since $\boldsymbol\Gamma(\mathcal J,\cU)$ is invertible, the seen symbols are independent and uniform. Gray: known to $\cT$; the shares of $k_0$ follow from those of $\cH'$ and the known row sums $-\sum_{k\in\cT}s_{t,k}$, and its noise column follows from the sum. In general, $\boldsymbol\Gamma(\mathcal J,\cU)$ has at least as many columns as rows.}
\label{fig:rank}
\end{figure}

\emph{Rank.} The coalition sees the shares $s_{1,4},s_{1,5},s_{2,4},s_{2,5}$, which are uniform given the inputs, and not the shares to parties $1,2,3$. We use an elementary fact: in a uniform additive sharing with a known sum, any shares except one are independent and uniform, since they determine the remaining one. Given the inputs and the seen shares, the shares $s_{t,1},s_{t,2},s_{t,3}$ of supplier $t\in[2]$ to the honest parties $1,2,3$ are uniform with the known sum $-(s_{t,4}+s_{t,5})$ (Figure~\ref{fig:rank}). Omitting party $2$, the share vectors $\boldsymbol s_1$ and $\boldsymbol s_3$ are independent and uniform in $\mathbb F_q^2$. For $k\in\{1,3\}$,
\[
(\tilde W_k(2),\tilde W_k(3))=(W_k(2),W_k(3))+\boldsymbol\Gamma(\{2,3\},[2])\boldsymbol s_k,
\]
and $\boldsymbol\Gamma(\{2,3\},[2])$ is invertible, so the four symbols $\tilde W_1(2),\tilde W_1(3),\tilde W_3(2),\tilde W_3(3)$ are independent and uniform, also given the seen shares. Hence $H(Y_{\cT}\,|\,W_{[5]})\ge H(Y'_{\cT}\,|\,W_{[5]})=4+4=8$, the $4$ seen shares and the $4$ noisy input symbols, which matches the count, and the mutual information is $0$, which is~\eqref{eq:security}.

Every other coalition is handled in the same way, and only the submatrix of $\boldsymbol\Gamma$ changes. Each party is either a supplier or a collector, so each member of a coalition either removes a column from $\boldsymbol\Gamma$, if it is a supplier, since it knows its own shares, or adds a row, if it is a collector, the sum that it collects. Starting from the $T=2$ columns and no rows, a coalition of at most $2$ parties leaves a submatrix $\boldsymbol\Gamma(\mathcal J,\cU)$ with at least as many columns as rows, which has full row rank. For example, $\cT=\{1,3\}$ gives $\boldsymbol\Gamma(\{1\},\{2\})$ and $\cT=\{3,4\}$, which does not contain the destination, gives $\boldsymbol\Gamma(\{1,2\},[2])$; for a coalition without the destination, the omitted party is $k_0=5$, which sends no noisy input symbols.

\emph{Rate.} The shares take $8$ symbols, and each of the three sums takes $4$, one from every party other than the destination (Figure~\ref{fig:scheme}). Hence the scheme sends $20$ symbols for $L=3$, $2$ on every link; party $1$ uses $5$ randomness symbols and party $2$ uses $3$. Thus $R_X=20/3$ and $(R_{Z_1},\ldots,R_{Z_5})=(5/3,1,0,0,0)$, an extreme point of Theorem~\ref{thm:main}.
\end{example}

\subsection{A Scheme at Every Extreme Point}\label{sec:scheme-main}

Fix an extreme point~\eqref{eq:extreme}: distinct parties $p_1,\ldots,p_T$, the suppliers, and the other parties $c_1,\ldots,c_{K-T}$, the collectors, in any order. First let $q\ge K$ and $L=K-T$. For distinct $\alpha_1,\ldots,\alpha_{K-T},\beta_1,\ldots,\beta_T\in\Fq$, let $\boldsymbol\Gamma$ be the $(K-T)\times T$ Cauchy matrix with entries $\boldsymbol\Gamma(\ell,t)=1/(\alpha_\ell-\beta_t)$, whose square submatrices are all invertible, since each is again a Cauchy matrix and has a nonzero determinant.

\emph{Shares and noise.} Supplier $p_t$, $t\in[T]$, deals an additive sharing of zero $s_{t,k}$, $k\in[K]$: the shares $s_{t,k}$, $k\ne p_t$, are independent and uniform, and
\[
s_{t,p_t}=-\sum_{k\ne p_t}s_{t,k},\qquad\text{so that}\qquad\sum_{k\in[K]}s_{t,k}=0.
\]
Supplier $p_t$ draws its shares to the collectors and, for $t<t'$, both shares $s_{t,p_{t'}}$ and $s_{t',p_t}$ between $p_t$ and $p_{t'}$, and all drawn symbols are independent. Party $k$ holds the share vector $\boldsymbol s_k$, with one share from each of the $T$ sharings, and its noise $S_k$ has one symbol $S_k(\ell)$ for each sum $\ell\in[K-T]$:
\[
\boldsymbol s_k=(s_{1,k},\ldots,s_{T,k}),\qquad S_k=\boldsymbol\Gamma\boldsymbol s_k,\quad\text{that is,}\quad S_k(\ell)=\sum_{t\in[T]}\boldsymbol\Gamma(\ell,t)\,s_{t,k}.
\]
Since the shares of every supplier add up to zero, $\sum_{k\in[K]}S_k=0$. Let $\tilde W_k=W_k+S_k$ be the noisy input of party $k$, with symbols $\tilde W_k(\ell)$, $\ell\in[K-T]$.

\emph{Rounds.} The destination is party $K$, and sum $\ell$ is collected at $c_\ell$.
\begin{itemize}
\item Round 1 (phase 1): supplier $p_t$ sends $s_{t,c_\ell}$ to $c_\ell$, $\ell\in[K-T]$, and both $s_{t,p_{t'}}$ and $s_{t',p_t}$ to $p_{t'}$, $t'>t$. Then every party knows its share vector, hence its noise.
\item Round 2 (phase 2): for every $\ell\in[K-T]$, every party $k\notin\{K,c_\ell\}$ sends $\tilde W_k(\ell)$ to $c_\ell$.
\item Round 3 (phase 2): every collector $c_\ell\ne K$ sends the partial sum $P_\ell=\sum_{k\ne K}\tilde W_k(\ell)$ of its own symbol $\tilde W_{c_\ell}(\ell)$ and those it received to $K$.
\end{itemize}

\emph{Correctness.} The destination outputs $W_{\Sigma}(\ell)=\tilde W_K(\ell)+P_\ell$ for every sum $\ell$ that it does not collect, and, if $K=c_{\ell'}$, $W_{\Sigma}(\ell')=\sum_{k\in[K]}\tilde W_k(\ell')$ from its own symbol $\tilde W_K(\ell')$ and those it received. Each output equals $\sum_k\tilde W_k(\ell)=\sum_kW_k(\ell)$ for its sum $\ell$, because the noise adds up to zero over the parties.

\emph{Security.} Consider a coalition $\cT$ with $|\cT|\le T$. Let $\mathcal J=\{\ell:c_\ell\in\cT\}$ index the sums collected in $\cT$, the only sums in which $\cT$ sees noisy input symbols, and $\cU=\{t:p_t\notin\cT\}$ index the honest suppliers, whose shares to the honest parties $\cT$ does not see. Let $\cH=[K]\setminus\cT$ be the honest parties, $|\cH|\ge K-T$, let $k_0\in\cH$ be the destination if $K\in\cH$, which sends no noisy input symbols, and otherwise a fixed honest party, whose noisy input symbols follow from the sum, and let $\cH'=\cH\setminus\{k_0\}$. Since every party is a supplier or a collector, each member of $\cT$ removes one of the $T$ columns of $\boldsymbol\Gamma$ or adds one row: $\cT$ contains $T-|\cU|$ suppliers and $|\mathcal J|$ collectors, so $|\cU|-|\mathcal J|=T-|\cT|\ge0$. Let $Y_{\cT}\triangleq(Y_k(1{:}N))_{k\in\cT}$ be the messages that $\cT$ receives, and let $V_{\cT}$ be the conditioning in~\eqref{eq:security}, $V_{\cT}=(W_{\cT},Z_{\cT},\Ws)$ if $K\in\cT$ and $V_{\cT}=(W_{\cT},Z_{\cT})$ otherwise. By the definition of mutual information,
\begin{equation}\label{eq:mi}
I\big(W_{[K]};Y_{\cT}\,\big|\,V_{\cT}\big)=H(Y_{\cT}\,|\,V_{\cT})-H\big(Y_{\cT}\,\big|\,W_{[K]},V_{\cT}\big).
\end{equation}

\emph{Counting.} Let $Y'_{\cT}$ consist of (i) the shares that $\cT$ receives from honest suppliers in round 1, and (ii) the noisy input symbols $\tilde W_k(\ell)$, $k\in\cH'$, $\ell\in\mathcal J$. Every symbol of $Y_{\cT}$ is in $Y'_{\cT}$ or a function of $V_{\cT}$ and $Y'_{\cT}$ (Table~\ref{tab:observation}), as the parties in $\cT$ know their share vectors, hence their noisy inputs. Hence $H(Y_{\cT}\,|\,V_{\cT})\le|Y'_{\cT}|$, the number of symbols in $Y'_{\cT}$, which is the number of shares in (i) plus $|\mathcal J|\,|\cH'|=|\mathcal J|(K-|\cT|-1)$ symbols in (ii).

\emph{Rank.} Given the inputs and $V_{\cT}$, the symbols $Y'_{\cT}$ are a linear function of the randomness of the honest suppliers, which is uniform and independent of the inputs and of $V_{\cT}$, plus a known vector, so the second term of~\eqref{eq:mi} is at least $H(Y'_{\cT}\,|\,W_{[K]},V_{\cT})$, the rank of this function. The shares (i) are themselves randomness symbols. Given these shares, for $t\in\cU$ the shares $s_{t,k}$, $k\in\cH\setminus\{p_t\}$, are unseen, independent and uniform, and $s_{t,p_t}$ completes the known sum $-\sum_{k\in\cT}s_{t,k}$, independently over $t\in\cU$; hence any $|\cH|-1$ of the shares $s_{t,k}$, $k\in\cH$, are independent and uniform. Omitting $k_0$, the vectors $(s_{t,k})_{t\in\cU}$, $k\in\cH'$, are independent and uniform in $\Fq^{|\cU|}$, and for $k\in\cH'$,
\[
\big(\tilde W_k(\ell)\big)_{\ell\in\mathcal J}=\big(W_k(\ell)\big)_{\ell\in\mathcal J}+\boldsymbol\Gamma(\mathcal J,\cU)\,(s_{t,k})_{t\in\cU}+\text{a known vector}.
\]
Since $|\cU|-|\mathcal J|=T-|\cT|$, the matrix $\boldsymbol\Gamma(\mathcal J,\cU)$ is square and invertible if $|\cT|=T$, and otherwise has more columns than rows, any $|\mathcal J|$ of which form an invertible matrix; in both cases it has full row rank, so the vectors $\boldsymbol\Gamma(\mathcal J,\cU)\,(s_{t,k})_{t\in\cU}$, $k\in\cH'$, are independent and uniform in $\Fq^{|\mathcal J|}$. Hence the shares (i), which are independent and uniform, contribute their number to the rank, and given them the $|\mathcal J|\,|\cH'|$ symbols (ii) contribute $|\mathcal J|\,|\cH'|$; the rank is $|Y'_{\cT}|$, the count above, the mutual information~\eqref{eq:mi} is $0$, and~\eqref{eq:security} holds for every coalition.

\emph{Rate.} Supplier $p_t$ draws $K-T$ shares for the collectors and $2(T-t)$ shares for the suppliers $p_{t'}$, $t'>t$, that is, $K+T-2t$ randomness symbols. Every link carries $2$ symbols (Table~\ref{tab:links}), so the scheme sends $2\binom K2=K(K-1)$ symbols. Dividing by $L=K-T$, $R_X=K(K-1)/(K-T)$ and the randomness rates are the extreme point~\eqref{eq:extreme}.

\begin{table}[t]
\centering\footnotesize
\renewcommand{\arraystretch}{1.2}
\begin{tabular}{lllll}
\toprule
Link & Round 1 & Round 2 & Round 3 & Total\\
\midrule
$p_t$ and $p_{t'}$, $t<t'$ & $s_{t,p_{t'}}$, $s_{t',p_t}$ & & & $2$\\
$p_t\ne K$ and $c_\ell$ & $s_{t,c_\ell}$ & $\tilde W_{p_t}(\ell)$ & & $2$\\
$p_t=K$ and $c_\ell$ & $s_{t,c_\ell}$ & & $P_\ell$ & $2$\\
$c_\ell$ and $c_{\ell'}$, both $\ne K$ & & $\tilde W_{c_\ell}(\ell')$, $\tilde W_{c_{\ell'}}(\ell)$ & & $2$\\
$c_\ell$ and $c_{\ell'}=K$ & & $\tilde W_{c_\ell}(\ell')$ & $P_\ell$ & $2$\\
\bottomrule
\end{tabular}
\caption{The symbols on each link in the scheme of Section~\ref{sec:scheme-main}, by the roles of its two endpoints; $K$ is the destination.}\label{tab:links}
\end{table}

\emph{Field.} If $q<K$, choose $B$ with $q^B\ge K$, run the scheme over $\F_{q^B}$, and represent each symbol of $\F_{q^B}$ by $B$ symbols of $\Fq$ in a fixed basis. Then $L=(K-T)B$, and all lengths and amounts of randomness scale by $B$, so the rates are unchanged, and every operation is $\Fq$-linear. This completes the proof of Theorem~\ref{thm:main}.\qed

\section{Generalizations}\label{sec:gen}

This section shows that the techniques of Sections~\ref{sec:converse} and~\ref{sec:scheme} settle two generalizations: output to every party (Section~\ref{sec:all}) and a server without input that computes the sum (Section~\ref{sec:server-sub}).

\subsection{Output to Every Party}\label{sec:all}

When every party recovers the sum, the problem is \emph{secure all-reduce}, the secure version of all-reduce, which synchronizes, for example, gradients in decentralized training~\cite{ZhouSun26}. This is also the convention of perfect $T$-privacy against passive coalitions in the works on private computation of the sum~\cite{BDPV99,KM97,GR05,GIS22}: every party recovers the sum, and every coalition may learn the sum but nothing more about the inputs. With every party recovering the sum, the optimal communication rate increases by exactly $K-1$, and the randomness constraints are unchanged.

\begin{theorem}\label{thm:all}
In the model of Section~\ref{sec:model}, let every party recover the sum, $H(\Ws\,|\,W_k,Z_k,Y_k(1{:}N))=0$ for all $k\in[K]$, and let every coalition learn nothing about the inputs beyond what it already knows and the sum:
\begin{equation}\label{eq:security-all}
I\big((W_k)_{k\in[K]};(Y_k(1{:}N))_{k\in\cT}\,\big|\,\Ws,(W_k,Z_k)_{k\in\cT}\big)=0,\quad\forall\cT\subseteq[K],\ |\cT|\le T.
\end{equation} For $K\ge2$ and $0\le T\le K-2$, the capacity region is the set of rate tuples $(R_X,R_{Z_1},\ldots,R_{Z_K})\in\mathbb R_{\ge0}^{K+1}$ such that
\[
R_X\ge\frac{K(K-1)}{K-T}+K-1
\]
and the randomness rates satisfy the constraints of~\eqref{eq:region}.
\end{theorem}

Hence it is optimal to compute the sum securely at one party and then send it to the other $K-1$ parties: interaction cannot save communication by letting the parties learn the sum along the way. The new step in the converse is a lower bound on communication by the correlation that the messages create among the final views $(\Omega_k,X_{\cE_k})$, $k\in[K]$, of the parties, measured by their total correlation $\sum_kH(\Omega_k,X_{\cE_k})-H(\Omega_{[K]},X)$, where $\Omega_k=(W_k,Z_k)$.

\begin{proof}
\emph{Converse: randomness.} The converse of Theorem~\ref{thm:main} uses only the encoding~\eqref{eq:message}, the recovery of the sum at the destination,~\eqref{eq:correct}, and the security condition given the sum,~\eqref{eq:security-sum}. A protocol in which every party recovers the sum satisfies~\eqref{eq:correct}. It also satisfies~\eqref{eq:security-sum}: the messages that the parties in $\cT$ send are functions of $(W_k,Z_k)_{k\in\cT}$ and of the messages they receive, so~\eqref{eq:security-all} remains true when the received messages are replaced by $X_{\cE(\cT)}$, all messages that the parties in $\cT$ send or receive, which is~\eqref{eq:security-sum}. Hence the randomness constraints of~\eqref{eq:region} hold.

\emph{Converse: communication.} Let $\bar\Omega_k=(\Omega_j)_{j\ne k}$ as in the proof of Lemma~\ref{lem:recover}. Three facts are used.
\begin{itemize}
\item By~\eqref{eq:message}, the messages that party $k$ sends are functions of $\Omega_k$ and of the messages it receives, so $H(X_{\cE_k}\,|\,\Omega_k)$ is at most the total length of the messages that party $k$ receives; summing over $k$ counts every message once, so $L_X\ge\sum_{k}H(X_{\cE_k}\,|\,\Omega_k)$.
\item Every party recovers the sum, so $\Ws$ is a function of $(\Omega_k,X_{\cE_k})$ for every $k$.
\item The first step of the proof of Lemma~\ref{lem:recover} gives $I(\Omega_k;\bar\Omega_k\,|\,X_{\cE_k})=0$, and the messages on the other links are a function of $(\bar\Omega_k,X_{\cE_k})$, so $I(\Omega_k;\bar\Omega_k\,|\,X)=0$ for every $k$. Hence $H(\Omega_{[K]}\,|\,X)=\sum_kH(\Omega_k\,|\,X)$ by the chain rule, and, as in the recovery step of that proof, $\Ws$, which parties $1$ and $2$ both recover, is a function of $X$.
\end{itemize}
Since $H(\Omega_{[K]},X)=H(\Omega_{[K]})=\sum_kH(\Omega_k)$, the bound of the first fact is the total correlation of the final views, and since $H(\Omega_k\,|\,X_{\cE_k},\Ws)\ge H(\Omega_k\,|\,X)$,
\begin{align*}
\sum_{k\in[K]}H(X_{\cE_k}\,|\,\Omega_k)&=\sum_{k\in[K]}H(\Omega_k,X_{\cE_k})-H(\Omega_{[K]},X)\\
&\ge\sum_{k\in[K]}\big(L+H(X_{\cE_k}\,|\,\Ws)+H(\Omega_k\,|\,X)\big)-\Big(L+H(X\,|\,\Ws)+\sum_{k\in[K]}H(\Omega_k\,|\,X)\Big)\\
&=(K-1)L+\sum_{k\in[K]}H(X_{\cE_k}\,|\,\Ws)-H(X\,|\,\Ws)\\
&\overset{(a)}{\ge}(K-1)L+H(X\,|\,\Ws)\overset{(b)}{\ge}(K-1)L+\frac{K(K-1)}{K-T}L,
\end{align*}
where (a) is~\eqref{eq:shearer} of Lemma~\ref{lem:cover} with $\cF=\{\cE_k:k\in[K]\}$, in which every link lies in exactly $D=2$ sets, and (b) is the bound on $H(X\,|\,\Ws)$ of Section~\ref{sec:converse-main}. With the first fact, $R_X\ge K(K-1)/(K-T)+K-1$.

\emph{Achievability.} At each extreme point~\eqref{eq:extreme}, run the scheme of Section~\ref{sec:scheme-main}, or for $T=0$ send every input to the destination, and let the destination send $\Ws$ to every other party. This adds $(K-1)L$ symbols and no randomness. The observation of every coalition gains only $\Ws$, and since the scheme satisfies~\eqref{eq:security-sum}, which conditions on the sum, it satisfies~\eqref{eq:security-all}. Time sharing gives the rest of the region, as in Section~\ref{sec:extreme}.
\end{proof}

\subsection{Server}\label{sec:server-sub}

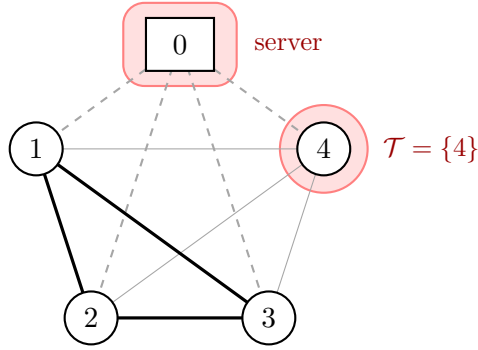
\begin{figure}[t]
\centering
\begin{tikzpicture}
\coordinate (s0) at (90:2.0);
\coordinate (u1) at (162:2.0); \coordinate (u2) at (234:2.0);
\coordinate (u3) at (306:2.0); \coordinate (u4) at (18:2.0);
\begin{scope}[on background layer]
  \fill[red!12,draw=red!50,thick,rounded corners=8pt] ($(s0)+(-0.75,-0.55)$) rectangle ($(s0)+(0.75,0.55)$);
  \fill[red!12,draw=red!50,thick] (u4) circle (0.58);
  \foreach \n in {1,2,3,4}{\draw[up link] (u\n)--(s0);}
  \foreach \i/\j in {1/4,2/4,3/4}{\draw[thin link] (u\i)--(u\j);}
  \foreach \i/\j in {1/2,1/3,2/3}{\draw[hon link] (u\i)--(u\j);}
\end{scope}
\foreach \n in {1,2,3,4}{\node[party] at (u\n) {$\n$};}
\node[serv] at (s0) {$0$};
\node[red!60!black,font=\small,anchor=west] at ($(s0)+(0.85,0)$) {server};
\node[red!60!black,font=\small,anchor=west] at ($(u4)+(0.65,0)$) {$\mathcal T=\{4\}$};
\end{tikzpicture}
\caption{The server model with $(K,T)=(4,1)$. Square: server; solid: links between the users; dashed: links between the users and the server; shaded: the server and a coalition $\cT$ of users, which collude together; thick: the links $\Eu[\cH]$ inside the honest set $\cH=\{1,2,3\}$.}
\label{fig:server}
\end{figure}

Let the sum be computed at a server without input, which may collude with at most $T$ of the $K$ \emph{users}, the input holders. This is close to the model of Section~\ref{sec:model} with a destination that has no input and belongs to every coalition, and the proofs carry over. The model differs from Section~\ref{sec:model} as follows. The server, indexed $0$, holds local randomness $Z_0$ of $L_{Z_0}$ i.i.d.\ uniform symbols and no input, and~\eqref{eq:indep} and~\eqref{eq:size} hold with $Z_0$ added. The secure links are $\Eu=\{\{i,j\}:i,j\in[K],\,i\ne j\}$ among the users and $\Es=\{\{0,k\}:k\in[K]\}$ between the users and the server, $\Eu[\cA]$ and $\Eu(\cA)$ are defined as in Section~\ref{sec:model} within $\Eu$, $\cE_k$ is the set of links at user $k$, including $\{0,k\}$, and messages satisfy~\eqref{eq:message} over $\{0\}\cup[K]$, with $Y_j(n)$ including the messages from the server; Figure~\ref{fig:server} shows an instance. The server recovers the sum, $H(\Ws\,|\,Z_0,Y_0(1{:}N))=0$, and the server together with at most $T$ users learns nothing beyond the sum:
\begin{equation}\label{eq:security-server}
I\big((W_k)_{k\in[K]};Y_0(1{:}N),(Y_k(1{:}N))_{k\in\cT}\,\big|\,\Ws,(W_k,Z_k)_{k\in\cT},Z_0\big)=0,\qquad\forall\cT\subseteq[K],\ |\cT|\le T.
\end{equation}
As in Section~\ref{sec:model}, the received messages in~\eqref{eq:security-server} can be replaced by $(X_{\Es},X_{\Eu(\cT)})$, all messages that the server and the users in $\cT$ send or receive. Since $Z_0$ is independent of the inputs and the users' local randomness, \eqref{eq:security-server} also implies that at most $T$ users without the server learn nothing about the inputs beyond the sum, $I((W_k)_{k\in[K]};(Y_k(1{:}N))_{k\in\cT}\,|\,\Ws,(W_k,Z_k)_{k\in\cT})=0$; they may learn the sum itself, which is harmless in federated learning, where the aggregate is returned to all users~\cite{Bonawitz17,Sun23}.
With $L_{\mathrm u}$ and $L_0$ the total lengths of the messages on $\Eu$ and on $\Es$, the \emph{user-to-user} and the \emph{server communication rates} are $R_{\mathrm u}\triangleq L_{\mathrm u}/L$ and $R_0\triangleq L_0/L$, the randomness rates are $R_{Z_k}\triangleq L_{Z_k}/L$, $k\in\{0\}\cup[K]$, and the capacity region is defined as in Section~\ref{sec:model}.

\begin{theorem}\label{thm:server}
For $K\ge2$ and $0\le T\le K-2$, the capacity region of interactive secure summation with a server is the set of rate tuples $(R_{\mathrm u},R_0,R_{Z_0},R_{Z_1},\ldots,R_{Z_K})\in\mathbb R_{\ge0}^{K+3}$ such that
\[
\begin{aligned}
R_{\mathrm u}&\ge\frac{K(K-1)}{K-T},\qquad R_0\ge1,\\
\sum_{k\in[K]\setminus\cA}R_{Z_k}&\ge\frac{(T-|\cA|)(K-|\cA|-1)}{K-T},\qquad\forall\,\cA\subseteq[K],\ |\cA|\le T.
\end{aligned}
\]
\end{theorem}

The user-to-user communication and the randomness of the users are as in Theorem~\ref{thm:main}, the server receives only the sum, and the randomness of the server is not constrained: the server belongs to every coalition, so its randomness never helps. In particular, the optimal total randomness rate $T(K-1)/(K-T)$ and the optimal server communication rate $1$ are smaller by $K-1$ each than the optimal rates of the one-hop model of~\cite{YZWT26}, in which the users generate correlated randomness through input-independent communication and then each user sends one message directly to the server.

\begin{proof}
\emph{Converse.} The proof of Section~\ref{sec:converse-main} applies with $Z_0$ and $X_{\Es}$ added to every conditioning, since the server knows them in every coalition. The proof of Lemma~\ref{lem:recover} applies with the server as the destination, which is never the user $k$, so $H(W_k\,|\,X_{\cE_k},\Ws)=0$ for every user $k$. With~\eqref{eq:security-server}, the coalition bound becomes $H(X_{\Eu[\cH]}\,|\,X_{\Es},X_{\Eu(\cT)},\Ws,Z_0)\ge(K-T-1)L$, the covering gives $H(X_{\Eu}\,|\,X_{\Es},\Ws,Z_0)\ge K(K-1)L/(K-T)$, hence the bound on $R_{\mathrm u}$, and the randomness outside every set $\cA$ of $M\le T$ users, including $\cA=\emptyset$, follows with $Q=(X_{\Es},X_{\Eu(\cA)},\Ws,W_{\cA},Z_{\cA},Z_0)$, where $H((W_k,Z_k)_{k\notin\cA}\,|\,W_{\cA},Z_{\cA},Z_0)=(K-M)L+\sum_{k\notin\cA}L_{Z_k}$ does not contain $L_{Z_0}$. Finally, the server recovers the sum, so $L_0\ge H(X_{\Es}\,|\,Z_0)\ge H(\Ws\,|\,Z_0)=L$.

\emph{Achievability.} At an extreme point of the polyhedron of randomness rates of the users, with user $K$ as the destination, the users run the scheme of Section~\ref{sec:scheme-main}, or for $T=0$ send their inputs to user $K$, and user $K$ sends $\Ws$ to the server; the server has no randomness. The server and a coalition $\cT$ observe what $\cT$ observes in that scheme, together with $\Ws$, and the scheme satisfies~\eqref{eq:security-sum}, which conditions on the sum, so~\eqref{eq:security-server} holds. The rates are $R_{\mathrm u}=K(K-1)/(K-T)$, $R_0=1$ and the extreme point, and time sharing gives the rest of the region.
\end{proof}

\section{Conclusion}\label{sec:conclusion}

We characterized the capacity region of interactive secure summation, where every pair of parties is connected by a secure link and each party has only independent local randomness, for all $K$ and $0\le T\le K-2$, with the communication and the randomness of every party, together with two generalizations. Table~\ref{tab:summary} collects the results. In each model the converse holds for all interactive protocols, and the optimal rates are achieved simultaneously by linear schemes.

\begin{table}[H]
\centering
\small\renewcommand{\arraystretch}{1.3}\setlength{\tabcolsep}{10pt}
\begin{tabular}{lll}
\toprule
Model & Communication & Randomness\\
\midrule
One destination (Thm.~\ref{thm:main}) & $R_X\ge K(K-1)/(K-T)$ & $\sum_{k\notin\cA}R_{Z_k}\ge\frac{(T-|\cA|)(K-|\cA|-1)}{K-T}$\\
\addlinespace
Every party (Thm.~\ref{thm:all}) & $R_X\ge K(K-1)/(K-T)+K-1$ & as for one destination\\
\addlinespace
Server (Thm.~\ref{thm:server}) & $R_{\mathrm u}\ge K(K-1)/(K-T)$ & as for one destination, users only\\
 & $R_0\ge1$ & \\
\bottomrule
\end{tabular}
\caption{Capacity regions for $0\le T\le K-2$; $\cA$ ranges over the sets of at most $T$ parties, users in the server model.}\label{tab:summary}
\end{table}

\appendix

\section{Relation to Results in Theoretical Computer Science}\label{app:tcs}

The results on private computation of the sum in theoretical computer science mostly concern one input bit or one group element per party and count messages and random bits for one computation, whereas Theorem~\ref{thm:main} concerns the cost per symbol of the sum as $L$ grows, so the two are not directly comparable. In the two extreme cases, the scalar values equal the optimal rates: for $T=1$, $K$ messages~\cite{CK93} and one random bit~\cite{KR98}, and for $T=K-2$, $\binom K2$ messages~\cite{CK93} and $\binom{K-1}2$ random bits~\cite{BDPV99}. Hence the scalar schemes are optimal there also for long inputs. A fair comparison is with the results for many instances computed together, which we discuss below. These results give the order of growth in $K$ and $T$ within their parameter regimes, while Theorem~\ref{thm:main} gives the exact optimal rates for every $K$ and $T$; in this sense the two study different regimes.

\emph{Communication.} Franklin and Yung~\cite{FY92} show for $K>2T+1$ and $T$ proportional to $K$ that one binary addition requires a number of bits of order $K^2$ while $K-T$ additions in parallel can be computed with $O(K^2\log K)$ bits, by packing several secrets into one polynomial. Over a binary base field with $B=\lceil\log_2K\rceil$, a scheme of Section~\ref{sec:scheme-main} computes $(K-T)B$ binary additions with $K(K-1)B$ bits, that is, $K(K-1)$ bits per $K-T$ binary additions, and with every party recovering the sum (Theorem~\ref{thm:all}) $(K-1)(2K-T)$ bits, both $O(K^2)$; the $\log K$ factor of~\cite{FY92} comes from placing one bit in each symbol of a field of size about $K$.

\emph{Randomness.} Goyal, Ishai and Song~\cite{GIS22} show how to amortize randomness over several additions. Their batch construction uses distinct parties to supply the input strings of an extractor, so its batch size is limited by the number of parties, and its cost keeps a batch overhead before reaching the amortized order $O(T\log^2K)$ per addition. With every party recovering the sum, Theorem~\ref{thm:all} gives the optimal total randomness rate $T(K-1)/(K-T)$ for every $K$ and $T$, which is asymptotic to $T$ as $T/K\to0$. Their randomness measure is the number of random bits, which over a binary base field is the number of randomness symbols counted here; the schemes of Section~\ref{sec:scheme-main} over a binary base field use exactly $T(K-1)B$ random bits for $(K-T)B$ binary additions. For $1$-private XOR of $L$ instances, Kushilevitz, Ostrovsky and Ros\'en~\cite{KOR03} give the lower bound $(1-2/K)L$ random bits; Theorem~\ref{thm:all} with $T=1$ and $q=2$ gives the optimal total randomness rate $T(K-1)/(K-T)=1$, that is, the larger lower bound of $L$ random bits, which the schemes of Section~\ref{sec:scheme-main} achieve.\footnote{Their model is slightly different, for example, the number of random bits that each party uses may depend on the execution, with at most $r$ in total; our converse also applies to it, since all the messages can then be generated from the inputs and $r$ independent uniform bits.}

\bibliography{refs}

\end{document}